\documentclass[11pt,twoside]{article}
\usepackage{amsmath,amssymb,latexsym,graphicx,hyperref,epstopdf,subcaption,color}
\usepackage{algpseudocode,algorithm,epstopdf}
\usepackage[capitalise]{cleveref}
\normalsize
\newtheorem{definition}{Definition}

\newtheorem{theorem}{Theorem}
\newtheorem{lemma}{Lemma}

\newtheorem{operation}{Operation}

\newenvironment{proof}{{\sc Proof. }}{\hfill$\Box$\vspace{0.2in}}

\newcommand{\Max}[1]{{\sc MaxP$^{{#1}+}$PC}}

\def\mcP{\mathcal{P}}
\def\mcQ{\mathcal{Q}}

\title{Approximation algorithms for covering vertices by long paths%
\footnote{An extended abstract appears in the {\em Proceedings of Conference MFCS 2022}.}}
\author{
	Mingyang~Gong\thanks{Department of Computing Science, University of Alberta.  Edmonton, Canada.
	\texttt{\{mgong4, bedgar, jfan10, guohui\}@ualberta.ca}}
\and
	Brett~Edgar$^\dagger$
\and
	Jing~Fan$^\dagger$\thanks{College of Arts and Sciences, Shanghai Polytechnic University.  Shanghai, China.}
\and
	Guohui~Lin$^\dagger$\thanks{Correspondence author. \texttt{guohui@ualberta.ca}}
\and
	Eiji~Miyano\thanks{Department of Artificial Intelligence, Kyushu Institute of Technology. Iizuka, Japan.
	\texttt{miyano@ai.kyutech.ac.jp}}}%
\date{\today}
\begin{document}
\maketitle

\begin{abstract}
Given a graph, the general problem to cover the maximum number of vertices by a collection of vertex-disjoint long paths
seemingly escapes from the literature.
A path containing at least $k$ vertices is considered long.
When $k \le 3$, the problem is polynomial time solvable;
when $k$ is the total number of vertices, the problem reduces to the Hamiltonian path problem, which is NP-complete.
For a fixed $k \ge 4$, the problem is NP-hard and
the best known approximation algorithm for the weighted set packing problem implies a $k$-approximation algorithm.
To the best of our knowledge, there is no approximation algorithm directly designed for the general problem;
when $k = 4$, the problem admits a $4$-approximation algorithm which was presented recently.
We propose the first $(0.4394 k + O(1))$-approximation algorithm for the general problem and an improved $2$-approximation algorithm when $k = 4$.
Both algorithms are based on local improvement, and their theoretical performance analyses are done via amortization
and their practical performance is examined through simulation studies.

\paragraph{2012 ACM Subject Classification}{Theory of computation $\rightarrow$ Packing and covering problems}

\end{abstract}

\subsubsection*{Acknowledgments.}
This research is supported by a visiting scholarship from the Shanghai Polytechnic University (JF),
the NSERC Canada (BE, GL),
and the KAKENHI Grants JP21K11755 and JP17K00016 and JST CREST JPMJR1402 (EM).

\newpage
\section{Introduction}
Path Cover ({\sc PC}) is one of the most well-known NP-hard optimization problems in algorithmic graph theory~\cite{GJ79},
in which given a simple undirected graph $G = (V, E)$ one wishes to find a minimum collection of vertex-disjoint paths that cover all the vertices,
that is, every vertex of $V$ is on one of the paths.
It has numerous applications from the real life, such as transportation networks, communication networks and networking security.
In particular, it includes the Hamiltonian path problem~\cite{GJ79} as a special case,
which asks for the existence of a single path covering all the vertices.

The Hamiltonian path problem is NP-complete;
therefore, the {\sc PC} problem cannot be approximated within ratio $2$ if P $\ne$ NP.
In fact, to the best of our knowledge, there is no $o(|V|)$-approximation algorithm for the {\sc PC} problem.
In the literature, several alternative objective functions have been proposed and studied~\cite{BK06,AN07,PH08,AN10,RTM14,CCC18,GW20,CCL21}.
For example,
Berman and Karpinski~\cite{BK06} tried to maximize the number of edges on the paths in a path cover,
which is equal to $|V|$ minus the number of paths, and proposed a $7/6$-approximation algorithm.
Chen et al.~\cite{CCC18,CCL21} showed that finding a path cover with the minimum total number of order-$1$ and order-$2$ paths
(where the order of a path is the number of vertices on the path; i.e., singletons and edges) can be done in polynomial time,
but it is NP-hard to find a path cover with the minimum number of paths of order at most $\ell$ when $\ell \ge 3$.

Recently, Kobayashi et al.~\cite{KLM22} generalized the problem studied by Chen et al.~\cite{CCC18,CCL21}
to assign a weight representing its profit or cost to each order-$\ell$ path,
with the goal of finding a path cover of the maximum weight or the minimum weight, respectively.
For instance, when the weight $f(\ell)$ of an order-$\ell$ path is $f(\ell) = 1$ for any $\ell \le k$ and $f(\ell) = 0$ for any $\ell \ge k+1$,
where $k$ is a fixed integer,
the {\em minimization} problem reduces to the problem studied by Chen et al.~\cite{CCC18,CCL21};
when $f(\ell)$ is $f(\ell) = 1$ for any $\ell \le k$ and $f(\ell) = +\infty$ for any $\ell \ge k+1$,
the {\em minimization} problem is the so-called $k$-path partition problem \cite{YCH97,MT07,CGL19a,CGL19b,CGS19,CCK21a};
when $f(\ell)$ is $f(\ell) = 0$ for any $\ell \le |V| - 1$ but $f(|V|) \ne 0$,
the {\em maximization} problem reduces to the Hamiltonian path problem.

Given an integer $k \ge 4$, in the special case where $f(\ell) = 0$ for any $\ell < k$ but $f(\ell) = \ell$ for any $\ell \ge k$,
the {\em maximization} problem can be re-phrased as to find a set of vertex-disjoint paths of order at least $k$ to cover the most vertices,
denoted as \Max{k}.
The \Max{4} problem (i.e., $k = 4$) is complementary to finding a path cover with the minimum number of paths of order at most $3$~\cite{CCC18,CCL21},
and thus it is NP-hard.
Kobayashi et al.~\cite{KLM22} presented a $4$-approximation algorithm for \Max{4} by greedily adding an order-$4$ path or
extending an existing path to the longest possible.

For a fixed integer $k \ge 4$, the \Max{k} problem is NP-hard too~\cite{KLM22};
to the best of our knowledge there is no approximation algorithm designed directly for it.
Nevertheless, the \Max{k} problem can be cast as a special case of the Maximum Weighted $(2k-1)$-Set Packing problem~\cite{GJ79},
by constructing a set of $\ell$ vertices when they are traceable (that is, they can be formed into a path) in the given graph
and assigning its weight $\ell$, for every $\ell = k, k+1, \ldots, 2k-1$.
(This upper bound $2k-1$ will become clear in the next section.)
The Maximum Weighted $(2k-1)$-Set Packing problem is APX-complete~\cite{Hoc83} and
the best known approximation guarantee is $k - \frac 1{63,700,992} + \epsilon$ for any $\epsilon > 0$~\cite{Neu21}.

In this paper, we study the \Max{k} problem from the approximation algorithm perspective.
The problem and its close variants have many motivating real-life applications in various areas
such as various (communication, routing, transportation, optical etc.) network design~\cite{IMM05}.
For example, when a local government plans to upgrade its subway infrastructures,
the given map of rail tracks is to be decomposed into multiple disjoint lines of stations, each of which will be taken care of by a team of workers.
Besides being disjoint so that while some lines are under construction the other lines can function properly,
each line is expected long enough for the team to work on continuously during a shift without wasting time and efforts
to move themselves and materials from one point to another.
Viewing the map as a graph, the goal of planning is to find a collection of vertex-disjoint long paths to cover the most vertices
(and of course, possibly under some other real traffic constraints).

We contribute two approximation algorithms for the \Max{k} problem,
the first of which is a $(0.4394k + O(1))$-approximation algorithm for any fixed integer $k \ge 4$, denoted as {\sc Approx1}.
We note that {\sc Approx1} is the first approximation algorithm directly designed for the \Max{k} problem,
and it is a local improvement algorithm that iteratively applies one of the three operations, addition, replacement and double-replacement,
each takes $O(|V|^k)$ time and covers at least one more vertex.
While the addition and the replacement operations have appeared in the $4$-approximation algorithm for the \Max{4} problem~\cite{KLM22},
the double-replacement operation is novel and it replaces one existing path in the current path collection with two new paths.
At termination, that is, when none of the local improvement operations is applicable,
we show by an amortization scheme that each path $P$ in the computed solution
is attributed with at most $\rho(k) n(P)$ vertices covered in an optimal solution,
where $n(P)$ denotes the order of the path $P$ and $\rho(k) \le 0.4394k + 0.6576$ for any $k \ge 4$.

The second $O(|V|^8)$-time algorithm, denoted as {\sc Approx2}, is for the \Max{4} problem.
Besides the three operations in {\sc Approx1}, we design two additional operations, re-cover and look-ahead.
The re-cover operation aims to increase the number of $4$-paths in the solution,
and the look-ahead operation covers at least one more vertex by trying multiple paths equivalent to an existing path in the current solution
in order to execute a replacement operation.
With these two more local improvement operations,
we design a refined amortization scheme to show that, on average,
each vertex covered in the computed solution is attributed with at most two vertices covered in an optimal solution.
That is, {\sc Approx2} is a $2$-approximation algorithm for the \Max{4} problem.
We also show a lower bound of $\frac {16}9$ on the worst-case performance ratio of {\sc Approx2}.

The rest of the paper is organized as follows.
In Section 2, we introduce the basic notations and definitions. 
Section 3 is devoted to the \Max{k} problem, where we present the {\sc Approx1} algorithm and its performance analysis.
In Section 4, we present the {\sc Approx2} algorithm for the \Max{4} problem, and outline the performance analysis.
We conduct simulation studies to examine the practical performance of the two algorithms in Section 5,
where we discuss the graph instance generation scheme, the implementation specification, and the numerical results.
We conclude the paper in the last section with some possible future work.

\section{Preliminaries}
For a fixed integer $k \ge 4$, in the \Max{k} problem,
we are given a simple undirected graph and want to find a collection of vertex-disjoint paths of order at least $k$ to cover the maximum number of vertices.

We consider simple undirected graphs in this paper and we fix a graph $G$ for discussion.
Let $V(G)$ and $E(G)$ denote its vertex set and edge set in the graph $G$, respectively.
We simplify $V(G)$ and $E(G)$ as $V$ and $E$, respectively, when the underlying graph is clear from the context.
We use $n(G)$ to denote the {\em order} of $G$, that is, $n \triangleq n(G) = |V|$ is the number of vertices in the graph.
A subgraph $S$ of $G$ is a graph such that $V(S) \subseteq V(G)$ and $E(S) \subseteq E(G)$;
and likewise, $n(S) = |V(S)|$ denotes its order.
Given a subset of vertices $R \subseteq V$, the subgraph of $G$ {\em induced} on $R$ is denoted as $G[R]$,
of which the vertex set is $R$ and the edge set contains all the edges of $E$ each connecting two vertices of $R$.
A (simple) path $P$ in $G$ is a subgraph of which the vertices can be ordered into $(v_1, v_2, \ldots, v_{n(P)})$
such that $E(P) = \{\{v_i, v_{i+1}\}, i = 1, 2, \ldots, n(P)-1\}$.
A path of order $\ell$ is called an {\em $\ell$-path} (also often called a {\em length-$(\ell-1)$ path} in the literature).

In this paper we are most interested in paths of order at least $4$.
In the sequel, given an $\ell$-path $P$ with $\ell \ge 4$, we let $u_j$ denote the vertex of $P$ at distance $j$ from one ending vertex of $P$,
for $0 \le j \le \lceil\frac {\ell}2\rceil - 1$,
and $v_j$ denote the vertex of $P$ at distance $j$ from the other ending vertex, for $0 \le j \le \lfloor\frac {\ell}2\rfloor - 1$.
When $\ell$ is odd, then the center vertex of the path is $u_{\frac {\ell-1}2}$.
This way, a $(2s+1)$-path is represented as $u_0$-$u_1$-$\cdots$-$u_{s-1}$-$u_s$-$v_{s-1}$-$\cdots$-$v_1$-$v_0$,
and a $(2s)$-path is represented as $u_0$-$u_1$-$\cdots$-$u_{s-1}$-$v_{s-1}$-$\cdots$--$v_1$-$v_0$.
Though the path is undirected and the vertex naming is often arbitrary,
sometimes we will pick a particular endpoint of the path to be the vertex $u_0$.

Given another path $Q$, let $Q - P$ denote the subgraph of $Q$ by removing those vertices in $V(P)$, and the edges of $E(Q)$ incident at them, from $Q$.
Clearly, if $V(P) \cap V(Q) = \emptyset$, then $Q - P = Q$;
otherwise, $Q - P$ is a collection of sub-paths of $Q$ each has at least one endpoint that is adjacent to some vertex on $P$ through an edge of $E(Q)$.
For a collection $\mcP$ of vertex-disjoint paths, it is also a subgraph of $G$,
with its vertex set $V(\mcP) = \cup_{P \in \mcP} V(P)$ and edge set $E(\mcP) = \cup_{P \in \mcP} E(P)$.
We similarly define $Q - \mcP$ to be the collection of sub-paths of $Q$ after removing those vertices in $V(\mcP)$ from $V(Q)$,
together with the edges of $E(Q)$ incident at them.
Furthermore, for another collection $\mcQ$ of vertex-disjoint paths, we can define $\mcQ - \mcP$ analogously,
that is, $\mcQ - \mcP$ is the collection of sub-paths of the paths in $\mcQ$ after removing those vertices in $V(\mcP)$ from $V(\mcQ)$,
together with the edges of $E(\mcQ)$ incident at them.

\begin{definition}\label{def1}
{\em (Associatedness)}
Given two collections $\mcP$ and $\mcQ$ of vertex-disjoint paths,
if a path $S$ of $\mcQ - \mcP$ has an endpoint that is adjacent to a vertex $v$ in $V(\mcP) \cap V(\mcQ)$ in $\mcQ$,
then we say $S$ is {\em associated} with $v$.
\end{definition}
One sees that a path $S$ of $\mcQ - \mcP$ can be associated with zero to two vertices in $V(\mcP) \cap V(\mcQ)$,
and conversely, a vertex of $V(\mcP) \cap V(\mcQ)$ can be associated with zero to two paths in $\mcQ - \mcP$.

If the paths of the collection $\mcP$ all have order at least $k$,
then the vertices of $V(\mcP)$ are said {\em covered} by the paths of $\mcP$, or simply by $\mcP$.
Let $R = V - V(\mcP)$.
For any vertex $v \in V(\mcP)$,
{\em an extension} $e(v)$ at the vertex $v$ is a path in the subgraph $G[R]$ of $G$ induced on $R$ which has an endpoint adjacent to $v$ in $G$.
Note that there could be many extensions at $v$, and we use $n(v) = \max_{e(v)} n(e(v))$ to denote the order of the longest extensions at the vertex $v$.

\begin{lemma}
\label{lemma01}
Given two collections $\mcP$ and $\mcQ$ of vertex-disjoint paths in the graph $G$,
for every vertex $v \in V(\mcP)$, any path of $\mcQ - \mcP$ associated with $v$ has order at most $n(v)$.
\end{lemma}
\begin{proof}
The lemma holds since $\mcQ - \mcP$ is a subgraph of the induced subgraph $G[V - V(\mcP)]$,
that is, every path of $\mcQ - \mcP$ is a path in $G[V - V(\mcP)]$, and the associatedness (see Definition~\ref{def1}) is a special adjacency
through an edge of $E(\mcQ)$.
\end{proof}

Our goal is to compute a collection of vertex-disjoint paths of order at least $k$,
such that it covers the most vertices.
In our local improvement algorithms below, we start with the empty collection $\mcP = \emptyset$
to iteratively expand $V(\mcP)$ through one of a few operations, to be defined later.
Notice that for an $\ell$-path with $\ell \ge 2k$,
one can break it into a $k$-path and an $(\ell-k)$-path by deleting an edge.
Since they cover the same vertices, we assume without loss of generality hereafter that
any collection $\mcP$ inside our algorithms contains vertex-disjoint paths of order in between $k$ and $2k-1$, inclusive.

\section{A $(0.4394k + 0.6576)$-approximation algorithm for \Max{k}}
For a given integer $k \ge 4$,
the best known approximation algorithm for the Maximum Weighted $(2k-1)$-Set Packing problem leads to
an $O(n^{2k-1})$-time $(k - \frac 1{63,700,992} + \epsilon)$-approximation algorithm for the \Max{k} problem, for any $\epsilon > 0$~\cite{Neu21}.
In this section, we define three local improvement operations for our algorithm for the \Max{k} problem, denoted as {\sc Approx1}.
We show later that its time complexity is $O(n^{k+1})$ and its approximation ratio is at most $0.4394k + 0.6576$.

For the current path collection $\mcP$, if there is a path covering $k$-vertices outside of $V(\mcP)$,
then the following operation adds the $k$-path into $\mcP$.

\begin{operation}\label{op1}
For a $k$-path $P$ in the induced subgraph $G[V - V(\mcP)]$,
the {\em Add($P$)} operation adds $P$ to $\mcP$.
\end{operation}

Since finding a $k$-path in the induced subgraph $G[V - V(\mcP)]$, for any $\mcP$, can be done in $O(n^k)$ time,
determining whether or not an addition operation is applicable, and if so then applying it, can be done in $O(n^k)$ time too.
Such an operation increases $|V(\mcP)|$ by $k$.

Recall that a path $P \in \mcP$ is represented as $u_0$-$u_1$-$\cdots$-$v_1$-$v_0$.
Though it is undirected, we may regard $u_0$ the {\em head} vertex of the path and $v_0$ the {\em tail} vertex for convenience.
The next operation seeks to extend a path of $\mcP$ by replacing a prefix (or a suffix) with a longer one.

\begin{operation}\label{op2}
For a path $P \in \mcP$ such that there is an index $t$ and an extension $e(u_t)$ with $n(e(u_t)) \ge t + 1$
(an extension $e(v_t)$ with $n(e(v_t)) \ge t + 1$, respectively),
the {\em Rep($P$)} operation replaces the prefix $u_0$-$u_1$-$\cdots$-$u_{t-1}$ of $P$ by $e(u_t)$
(the suffix $v_{t-1}$-$\cdots$-$v_1$-$v_0$ of $P$ by $e(v_t)$, respectively).
\end{operation}

Similarly, one sees that finding an extension $e(u_t)$ (of order at most $k-1$, or otherwise an Add operation is applicable)
in the induced subgraph $G[V - V(\mcP)]$, for any vertex $u_t \in P \in \mcP$, can be done in $O(n^{k-1})$ time.
Therefore, determining whether or not a prefix or a suffix replacement operation is applicable, and if so then applying it,
can be done in $O(n^k)$ time.
Note that such an operation increases $|V(\mcP)|$ by at least $1$.

The third operation tries to use a prefix and a non-overlapping suffix of a path in $\mcP$ to grow them into two separate paths of order at least $k$.

\begin{operation}\label{op3}
For a path $P \in \mcP$ such that
\begin{itemize}
\parskip=0pt
\item[(i)]
	there are two indices $t$ and $j$ with $j \ge t+1$ and two vertex-disjoint extensions $e(u_t)$ and $e(u_j)$
	with $n(e(u_t)) \ge k - (t+1)$ and $n(e(u_j)) \ge k - (n(P) - j)$,
	the {\em DoubleRep($P$)} operation replaces $P$ by two new paths $P_1 = u_0$-$u_1$-$\cdots$-$u_t$-$e(u_t)$ and $P_2 = e(u_j)$-$u_j$-$\cdots$-$v_1$-$v_0$;
\item[(ii)]
	or there are two indices $t$ and $j$ and two vertex-disjoint extensions $e(u_t)$ and $e(v_j)$
	with $n(e(u_t)) \ge k - (t+1)$ and $n(e(v_j)) \ge k - (j+1)$,
	the {\em DoubleRep($P$)} operation replaces $P$ by two new paths $P_1 = u_0$-$u_1$-$\cdots$-$u_t$-$e(u_t)$ and $P_2 = e(v_j)$-$v_j$-$\cdots$-$v_1$-$v_0$.
\end{itemize}
\end{operation}

Note that finding an extension $e(u_t)$ (of order at most $t$, or otherwise a Rep operation is applicable)
can be limited to those indices $t \ge \frac {k-1}2$.
Furthermore, we only need to find an extension $e(u_t)$ of order at most $\frac {k-1}2$
(equal to $\frac {k-1}2$ only if $t = \frac {k-1}2$).
For the same reason, we only need to find an extension $e(v_j)$ of order at most $\frac {k-1}2$ for those indices $j \ge \frac {k-1}2$
(equal to $\frac {k-1}2$ only if $j = \frac {k-1}2$).
Since $k \le n(P) \le 2k-1$ for any $P \in \mcP$,
we only need to find an extension $e(u_j)$ of order at most $\frac {k-1}2$ too
(equal to $\frac {k-1}2$ only if $j = \frac {n(P) - 1}2$ and $n(P) = k$).
In summary, finding the two vertex-disjoint extensions $e(u_t)$ and $e(u_j)$, or $e(u_t)$ and $e(v_j)$, 
in the induced subgraph $G[V - V(\mcP)]$ can be done in $O(n^{k-1})$ time
(in $\Theta(n^{k-1})$ for at most $2$ pairs of $t$ and $j$).
It follows that determining whether or not a double replacement operation is applicable, and if so then applying it,
can be done in $O(n^k)$ time.
Also, such an operation increases $|V(\mcP)|$ by at least $1$
as the total number of vertices covered by the two new paths $P_1$ and $P_2$ is at least $2k$.
We summarize the above observations on the three operations into the following lemma.

\begin{lemma}\label{lemma02}
Given a collection $\mcP$ of vertex-disjoint paths of order in between $k$ and $2k-1$,
determining whether or not one of the three operations {\em Add}, {\em Rep} and {\em DoubleRep} is applicable, and if so then applying it,
can be done in $O(n^k)$ time.
Each operation increases $|V(\mcP)|$ by at least $1$.
\end{lemma}

Given a graph $G = (V, E)$, our approximation algorithm for the \Max{k} problem, denoted as {\sc Approx1}, is iterative.
It starts with the empty collection $\mcP = \emptyset$;
in each iteration, it determines whether any one of the three operations Add, Rep and DoubleRep is applicable,
and if so then it applies the operation to update $\mcP$.
During the entire process, $\mcP$ is maintained to be a collection of vertex-disjoint paths of order in between $k$ and $2k-1$.
The algorithm terminates if none of the three operations is applicable for the current $\mcP$, and returns it as the solution.
A simple high level description of the algorithm {\sc Approx1} is depicted in Figure~\ref{fig01}.
From Lemma~\ref{lemma02}, we see that each operation improves the collection $\mcP$ to cover at least one more vertex.
Therefore, the overall running time of {\sc Approx1} is in $O(n^{k+1})$.

\begin{figure}[htb]
\begin{center}
\framebox{
\begin{minipage}{5.1in}
Algorithm {\sc Approx1}:\\
Input: A graph $G = (V, E)$;
\begin{itemize}
\parskip=0pt
\item[1.]
	initialize $\mcP = \emptyset$;
\item[2.]
	{\bf while} (one of the operations Add, Rep, DoubleRep is applicable)
	\begin{itemize}
	\parskip=0pt
	\item[2.1]
		apply the operation to update $\mcP$;
	\item[2.2]
		break any path of order $2k$ or above into two paths, one of which is a $k$-path;
	\end{itemize}
\item[3.]
	return the final $\mcP$.
\end{itemize}
\end{minipage}}
\end{center}
\caption{A high level description of the algorithm {\sc Approx1}.\label{fig01}}
\end{figure}

Below we fix $\mcP$ to denote the collection of paths returned by our algorithm {\sc Approx1}.
The next three lemmas summarize the structural properties of $\mcP$,
which are useful in the performance analysis.

\begin{lemma}\label{lemma03}
For any path $Q$ of order at least $k$ in the graph $G$, $V(Q) \cap V(\mcP) \ne \emptyset$.
\end{lemma}
\begin{proof}
The lemma holds due to the termination condition of the algorithm {\sc Approx1}, since otherwise an Add operation is applicable.
\end{proof}

\begin{lemma}\label{lemma04}
For any path $P \in \mcP$, $n(u_j) \le j$ and $n(v_j) \le j$ for any index $j$.
\end{lemma}
\begin{proof}
The lemma holds due to the termination condition of the algorithm {\sc Approx1}, since otherwise a Rep operation is applicable.
\end{proof}

\begin{lemma}\label{lemma05}
Suppose there is a vertex $u_t$ on a path $P \in \mcP$ and an extension $e(u_t)$ with $n(e(u_t)) \ge k-t-1$.
Then,
\begin{itemize}
\parskip=0pt
\item[(i)]
	for any vertex $u_j$ with $j \ge t+1$, $j \le k-2$ and every extension $e(u_j)$ vertex-disjoint to $e(u_t)$ has order $n(e(u_j)) \le k-j-2$;
\item[(ii)]
	for any index $j$, every extension $e(v_j)$ vertex-disjoint to $e(u_t)$ has order $n(e(v_j)) \le k-j-2$.
\end{itemize}
\end{lemma}
\begin{proof}
The lemma holds due to the termination condition of the algorithm {\sc Approx1}.

First, for any vertex $u_j$ with $j \ge t+1$, an extension $e(u_j)$ vertex-disjoint to $e(u_t)$ has order $n(e(u_j)) \le k - (n(P) - j) - 1$
since otherwise a DoubleRep operation is applicable.
Using the fact that $n(P) \ge 2j+1$, $n(e(u_j)) \le k - j - 2$.
Then by $k - j - 2 \ge 0$, we have $j \le k - 2$.

Next, similarly, for any index $j$, a vertex-disjoint extension $e(v_j)$ to $e(u_t)$ has order $n(e(v_j)) \le k - j - 2$
since otherwise a DoubleRep operation is applicable.
\end{proof}

We next examine the performance of the algorithm {\sc Approx1}.
We fix $\mcQ$ to denote an optimal collection of vertex-disjoint paths of order at least $k$ that covers the most vertices.
We apply an amortization scheme to {\em assign} the vertices of $V(\mcQ)$ to the vertices of $V(\mcP) \cap V(\mcQ)$.
We will show that, using the structural properties of $\mcP$ in Lemmas~\ref{lemma03}--\ref{lemma05},
the average number of vertices {\em received} by a vertex of $V(\mcP)$ is upper bounded by $\rho(k)$,
which is the approximation ratio of {\sc Approx1}.

In the amortization scheme, we assign the vertices of $V(\mcQ)$ to the vertices of $V(\mcP) \cap V(\mcQ)$ as follows:
Firstly, assign each vertex of $V(\mcQ) \cap V(\mcP)$ to itself.
Next, recall that $\mcQ - \mcP$ is the collection of sub-paths of the paths of $\mcQ$ after removing those vertices in $V(\mcQ) \cap V(\mcP)$.
By Lemma~\ref{lemma03}, each path $S$ of $\mcQ - \mcP$ is associated with one or two vertices in $V(\mcQ) \cap V(\mcP)$.
If the path $S$ is associated with only one vertex $v$ of $V(\mcQ) \cap V(\mcP)$,
then all the vertices on $S$ are assigned to the vertex $v$.
If the path $S$ is associated with two vertices $v_1$ and $v_2$ of $V(\mcQ) \cap V(\mcP)$,
then a half of the vertices on $S$ are assigned to each of the two vertices $v_1$ and $v_2$.
One sees that in the amortization scheme,
all the vertices of $V(\mcQ)$ are assigned to the vertices of $V(\mcP) \cap V(\mcQ)$;
conversely, each vertex of $V(\mcP) \cap V(\mcQ)$ receives itself, plus some fraction of or all the vertices on one or two paths of $\mcQ - \mcP$.
(We remark that the vertices of $V(\mcP) - V(\mcQ)$, if any, receive nothing.)

\begin{lemma}\label{lemma06}
For any vertex $u_j$ on a path $P \in \mcP$ with $j \le k-2$,
if $n(u_j) \le k - j - 2$, then $u_j$ receives at most $\frac 32 \min\{j, k-j-2\} + 1$ vertices.
\end{lemma}
\begin{proof}
By Lemma~\ref{lemma04}, we have $n(u_j) \le j$.
Therefore, $n(u_j) \le \min\{j, k-j-2\}$.
By Lemma~\ref{lemma01}, any path of $\mcQ - \mcP$ associated with $u_j$ contains at most $\min\{j, k-j-2\}$ vertices.

If there is at most one path of $\mcQ - \mcP$ associated with $u_j$, then the lemma is proved.

Consider the remaining case where there are two paths of $\mcQ - \mcP$ associated with $u_j$.
Since $2\min\{j, k-j-2\} + 1 \le j + (k-j-2) + 1 = k - 1$,
while the path $Q \in \mcQ$ containing the vertex $u_j$ has order at least $k$,
we conclude that one of these two paths of $\mcQ - \mcP$ is associated with another vertex of $V(\mcQ) \cap V(\mcP)$.
It follows from the amortization scheme that $u_j$ receives at most $\frac 32 \min\{j, k-j-2\} + 1$ vertices.
\end{proof}

\begin{lemma}
\label{lemma07}
Suppose $s$ is an integer such that $\lfloor\frac k2\rfloor - 1 \le s \le k-2$. 
Then we have 
\[
\sum_{j=1}^s \min\{j, k - j - 2\} = ks + \frac 32 k - \frac 14 k^2 - \frac 52 s - \frac 12 s^2 - 2 - \frac 14 (k\bmod{2}),
\]
where $\bmod$ is the modulo operation.
\end{lemma}
\begin{proof}
The formula can be directly validated by distinguishing the two cases where $k$ is even or odd,
and using the fact that $\min\{j, k - j - 2\} = j$ if and only if $j \le \lfloor \frac k2 \rfloor - 1$.
\end{proof}

\begin{theorem}\label{thm01}
The algorithm {\sc Approx1} is an $O(|V|^{k+1})$-time $\rho(k)$-approximation algorithm for the \Max{k} problem, where $k \ge 4$ and
\[
\rho(k) = \left\{
\begin{array}{ll}
\frac {3k+1}{2} - \frac 14 \sqrt{18 k^2 - 3}, 	&\mbox{ if $k$ is odd};\\
\frac {3k+1}{2} - \frac 14 \sqrt{18 k^2 - 21},	&\mbox{ if $k$ is even}.
\end{array}\right.
\]
In particular, {\sc Approx1} is a $2.4$-approximation algorithm for \Max{4} and the ratio $2.4$ is tight.
\end{theorem}
\begin{proof}
Recall that all the vertices of $V(\mcQ)$ are assigned to the vertices of $V(\mcQ) \cap V(\mcP)$, through our amortization scheme.
Below we estimate for any path $P \in \mcP$ the total number of vertices received by the vertices of $V(P) \cap V(\mcQ)$, denoted as $r(P)$,
and we will show that $\frac {r(P)}{n(P)} \le \rho(k)$.

We fix a path $P \in \mcP$ for discussion.
If it exists, we let $t$ denote the smallest index $j$
such that the vertex $u_j$ on the path $P$ is associated with a path $e(u_j)$ of $\mcQ - \mcP$ with order $n(e(u_j)) \ge k-j-1$.
Note that if necessary we may rename the vertices on $P$,
so that the non-existence of $t$ implies any path of $\mcQ - \mcP$ associated with the vertex $u_j$ or $v_j$ has order at most $k-j-2$, for any index $j$.
Furthermore, by Lemma~\ref{lemma05}, if $t$ exists,
then any path of $\mcQ - \mcP$, except $e(u_t)$, associated with the vertex $u_j$ or $v_j$ has order at most $k-j-2$, for any index $j \ne t$.
We remark that $e(u_t)$ could be associated with another vertex on the path $P$.
When $n(P) = 2k-1$, $t$ exists and $t = k-1$.

We distinguish three cases for $n(P)$ based on its parity and on whether it reaches the maximum value $2k-1$.

Case 1. $n(P) = 2s + 1$ where $\frac {k-1}2 \le s \le k-2$.

Since $s \ge \frac {k-1}2$, we have $\min\{s, k-s-2\} = k-s-2$.
If the index $t$ does not exist, that is, any path of $\mcQ - \mcP$ associated with the vertex $u_j$ or $v_j$ has order at most $k-j-2$, for any index $j$,
then by Lemma~\ref{lemma06} each of the vertices $u_j$ and $v_j$ receives at most $\frac 32 \min\{j, k-j-2\} + 1$ vertices.
Hence we have
\begin{eqnarray}\label{eq1}
r(P) &\le& 2\sum_{j=0}^{s-1} \left(\frac 32 \min\{j, k-j-2\} + 1\right) + \frac 32 \min\{s, k-s-2\} + 1\nonumber\\
    &=& 3 \sum_{j=0}^s \min\{j, k-j-2\} - \frac 32 \min\{s, k-s-2\} + (2s + 1)\nonumber\\
    &\le& 3 \sum_{j=1}^s \min\{j, k-j-2\} - \frac 32 k + \frac 72 s + 4.
\end{eqnarray}

If the index $t$ exists, then by Lemmas~\ref{lemma04} and \ref{lemma01}, $n(e(u_t)) \le t$ and thus the vertex $u_t$ receives at most $2t+1$ vertices.
Note that if $e(u_t)$ is associated with another vertex on the path $P$,
then we count all the vertices of $e(u_t)$ towards $u_t$ but count none towards the other vertex (that is, we could overestimate $r(P)$).
It follows from Lemma~\ref{lemma06}, the above Eq.~(\ref{eq1}), $t \le s$, and $s \ge \frac {k-1}2$ that
\begin{eqnarray}\label{eq2}
r(P) &\le& 3 \sum_{j=1}^s \min\{j, k-j-2\} - \frac 32 k + \frac 72 s + 4 - \left(\frac 32\min\{t, k-t-2\} + 1\right) + (2t + 1)\nonumber\\
    &=& 3 \sum_{j=1}^s \min\{j, k-j-2\} - \frac 32 k + \frac 72 s + 4 + \left(2t - \frac 32\min\{t, k-t-2\}\right)\nonumber\\
    &=& 3 \sum_{j=1}^s \min\{j, k-j-2\} - \frac 32 k + \frac 72 s + 4 + \max\left\{\frac 12 t, \frac 72 t - \frac 32 k + 3\right\}\nonumber\\
    &\le& 3 \sum_{j=1}^s \min\{j, k-j-2\} - \frac 32 k + \frac 72 s + 4 + \max\left\{\frac 12 s, \frac 72 s - \frac 32 k + 3\right\}\nonumber\\
    &=& 3 \sum_{j=1}^s \min\{j, k-j-2\} - \frac 32 k + \frac 72 s + 4 + \left(\frac 72 s - \frac 32 k + 3\right)\nonumber\\
    &=& 3 \sum_{j=1}^s \min\{j, k-j-2\} - 3k + 7s + 7.
\end{eqnarray}

Combining Eqs.~(\ref{eq1}, \ref{eq2}), and by Lemma~\ref{lemma07}, in Case 1 we always have
\[
r(P) \le 3 \sum_{j=1}^s \min\{j, k-j-2\} - 3k + 7s + 7 = 3ks + \frac 32 k - \frac 34 k^2 - \frac 12 s - \frac 32 s^2 + 1 - \frac 34 (k\bmod{2}).
\]
Therefore, using $n(P) = 2s+1$ we have
\[
\frac {r(P)}{n(P)} \le 
\left\{
\begin{array}{ll}
\frac {3k+1}2 - \frac 38 (2s+1) - \frac {6k^2-1}{8(2s+1)} \le \frac {3k+1}2 - \frac 14 \sqrt{18k^2-3}, 	&\mbox{ if $k$ is odd};\\
\frac {3k+1}2 - \frac 38 (2s+1) - \frac {6k^2-7}{8(2s+1)} \le \frac {3k+1}2 - \frac 14 \sqrt{18k^2-21}, &\mbox{ if $k$ is even},
\end{array}\right.
\]
where the upper bound $\rho(k)$ is achieved when $n(P) \approx \sqrt{2} k$.

Case 2. $n(P) = 2k - 1$.

In this case, $t = k-1$ since $n(e(u_{k-1})) \ge 0$.
By Lemmas~\ref{lemma04} and \ref{lemma01}, the vertex $u_{k-1}$ receives at most $2k-1$ vertices.
It follows from Lemmas~\ref{lemma06} and \ref{lemma07} that
\begin{eqnarray*}
r(P) &\le& 2 \sum_{j=0}^{k-2} \left(\frac 32 \min\{j, k-j-2\} + 1\right) + (2k-1)\\
	&=& 3 \sum_{j=1}^{k-2} \min\{j, k-j-2\} + 4k - 3\\
	&=& \frac 34 k^2 + k - \frac 34 (k\bmod{2}).
\end{eqnarray*}  
Using $n(P) = 2k-1$, one can check that $\frac {r(P)}{n(P)} \le \frac 38 k + \frac {11}{14} < \rho(k)$
(that is, the ratio is strictly less than $\rho(k)$ for any $k \ge 4$).

Case 3. $n(P) = 2s$ where $\frac k2 \le s \le k-1$. 

Similar to Case 1, if the index $t$ does not exist, then by Lemma~\ref{lemma06} we have 
\begin{equation}\label{eq3}
r(P) \le 2\sum_{j=0}^{s-1} \left(\frac 32 \min\{j, k-j-2\} + 1\right)
	= 3 \sum_{j=0}^{s-1} \min\{j, k-j-2\} + 2s.
\end{equation} 

If the index $t$ exists, then by Lemmas~\ref{lemma04} and \ref{lemma01}, the vertex $u_t$ receives at most $2t+1$ vertices.
Again, note that if $e(u_t)$ is associated with another vertex on the path $P$,
then we count all the vertices of $e(u_t)$ towards $u_t$ but none to the other vertex.
Similarly as how we derive Eq.~(\ref{eq2}), it follows from Lemma~\ref{lemma06}, the above Eq.~(\ref{eq3}), $t \le s-1$, and $s - 1 \ge \frac k2 - 1$ that
\begin{eqnarray}\label{eq4}
r(P) &\le& 3 \sum_{j=1}^{s-1} \min\{j, k-j-2\} + 2s - \left(\frac 32\min\{t, k-t-2\} + 1\right) + (2t + 1)\nonumber\\
    &=& 3 \sum_{j=1}^{s-1} \min\{j, k-j-2\} + 2s + \left(2t - \frac 32\min\{t, k-t-2\}\right)\nonumber\\
    &=& 3 \sum_{j=1}^{s-1} \min\{j, k-j-2\} + 2s + \max\left\{\frac 12 t, \frac 72 t - \frac 32 k + 3\right\}\nonumber\\
    &\le& 3 \sum_{j=1}^{s-1} \min\{j, k-j-2\} + 2s + \max\left\{\frac 12 (s-1), \frac 72 (s-1) - \frac 32 k + 3\right\}\nonumber\\
    &=& 3 \sum_{j=1}^{s-1} \min\{j, k-j-2\} + 2s + \left(\frac 72 (s-1) - \frac 32 k + 3\right)\nonumber\\
    &=& 3 \sum_{j=1}^{s-1} \min\{j, k-j-2\} - \frac 32 k + \frac {11}2 s - \frac 12.
\end{eqnarray}

Combining Eqs.~(\ref{eq3}, \ref{eq4}), and by Lemma~\ref{lemma07} (using $\frac k2 - 1 \le s-1 \le k-2$), we always have
\[
r(P) \le 3 \sum_{j=1}^{s-1} \min\{j, k-j-2\} - \frac 32 k + \frac {11}2 s - \frac 12
	= 3ks - \frac 34 k^2 + s - \frac 32 s^2 - \frac 12 - \frac 34 (k\bmod{2}).
\]
Therefore, using $n(P) = 2s$ we have
\[
\frac {r(P)}{n(P)} \le 
\left\{
\begin{array}{ll}
\frac {3k+1}2 - \frac 34 s - \frac {3k^2+5}{8s} \le \frac {3k+1}2 - \frac 14 \sqrt{18k^2+30}, &\mbox{ if $k$ is odd};\\
\frac {3k+1}2 - \frac 34 s - \frac {3k^2+2}{8s} \le \frac {3k+1}2 - \frac 14 \sqrt{18k^2+12}, &\mbox{ if $k$ is even},
\end{array}\right.
\]
where the upper bound is achieved when $n(P) \approx \sqrt{2} k$.
Since this upper bound is strictly less than $\rho(k)$, we have $\frac {r(P)}{n(P)} < \rho(k)$ for any $k \ge 4$.

The above three cases on $n(P)$ together prove that the worst-case performance ratio of the algorithm {\sc Approx1} is at most $\rho(k)$,
for any $k \ge 4$.
Note that $\rho(k) \le \frac {3k+1}2 - \frac 14 \sqrt{18k^2 - 21} \le 0.4394 k  + 0.6576$.

When $k = 4$, that is, for the \Max{4} problem,
one can check that the largest ratio $\frac {r(P)}{n(P)}$ is achieved in Case 1 where $n(P) = 5$ and
$\frac {r(P)}{n(P)} = \frac {13}2 - \frac 38 \times 5 - \frac {89}{8 \times 5} = \frac {12}5$.
In other words, {\sc Approx1} is a $2.4$-approximation algorithm for \Max{4}.

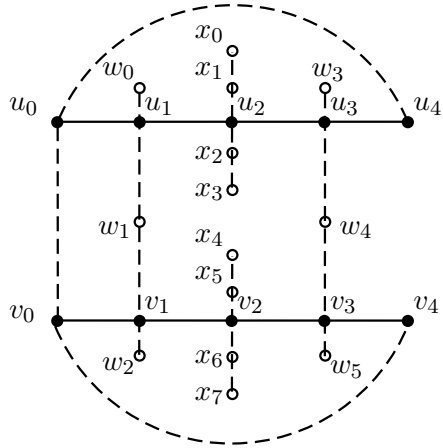
\begin{figure}[ht]
\centering
  \setlength{\unitlength}{1bp}%
  \begin{picture}(192.03, 176.64)(0,0)
  \put(0,0){\includegraphics{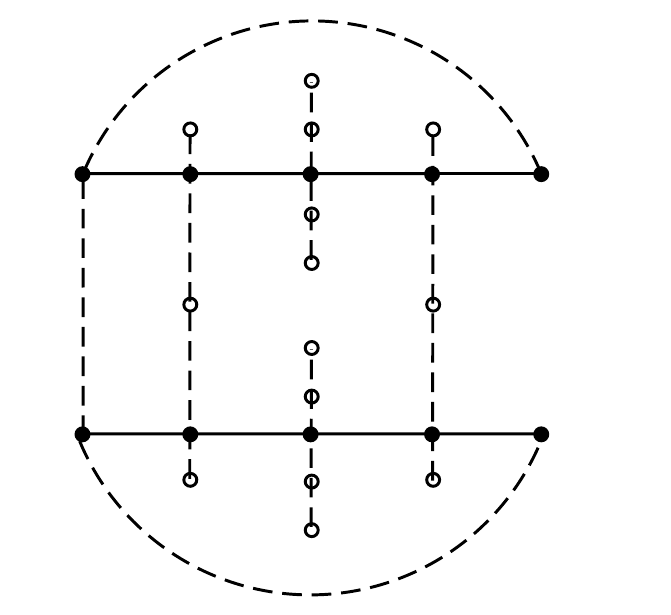}}
  \put(5.67,131.37){\fontsize{11.38}{13.66}\selectfont $u_0$}
  \put(56.62,131.37){\fontsize{11.38}{13.66}\selectfont $u_1$}
  \put(126.56,131.37){\fontsize{11.38}{13.66}\selectfont $u_3$}
  \put(5.67,53.18){\fontsize{11.38}{13.66}\selectfont $v_0$}
  \put(56.62,56.43){\fontsize{11.38}{13.66}\selectfont $v_1$}
  \put(126.56,56.43){\fontsize{11.38}{13.66}\selectfont $v_3$}
  \put(39.13,84.67){\fontsize{11.38}{13.66}\selectfont $w_1$}
  \put(40.64,144.61){\fontsize{11.38}{13.66}\selectfont $w_0$}
  \put(40.64,33.45){\fontsize{11.38}{13.66}\selectfont $w_2$}
  \put(158.02,131.37){\fontsize{11.38}{13.66}\selectfont $u_4$}
  \put(158.02,56.43){\fontsize{11.38}{13.66}\selectfont $v_4$}
  \put(91.59,131.37){\fontsize{11.38}{13.66}\selectfont $u_2$}
  \put(91.59,56.43){\fontsize{11.38}{13.66}\selectfont $v_2$}
  \put(75.60,158.60){\fontsize{11.38}{13.66}\selectfont $x_0$}
  \put(75.60,144.61){\fontsize{11.38}{13.66}\selectfont $x_1$}
  \put(75.60,113.13){\fontsize{11.38}{13.66}\selectfont $x_2$}
  \put(75.60,99.14){\fontsize{11.38}{13.66}\selectfont $x_3$}
  \put(75.60,81.65){\fontsize{11.38}{13.66}\selectfont $x_4$}
  \put(75.60,67.66){\fontsize{11.38}{13.66}\selectfont $x_5$}
  \put(75.60,36.19){\fontsize{11.38}{13.66}\selectfont $x_6$}
  \put(75.60,22.20){\fontsize{11.38}{13.66}\selectfont $x_7$}
  \put(130.07,84.67){\fontsize{11.38}{13.66}\selectfont $w_4$}
  \put(119.57,144.12){\fontsize{11.38}{13.66}\selectfont $w_3$}
  \put(126.57,32.20){\fontsize{11.38}{13.66}\selectfont $w_5$}
  \end{picture}%
\caption{A graph of order $24$ to show that the performance ratio $2.4$ of the algorithm {\sc Approx1} for \Max{4} is tight.
	All the edges in the graph are shown, either solid or dashed.
	The $10$ filled vertices are covered by the collection of two $5$-paths computed by {\sc Approx1},
	and the edges on these paths are shown solid;
	the edges on an optimal collection of four $5$-paths and a $4$-path, which covers all the vertices, are shown dashed.\label{fig02}}
\end{figure}

Furthermore, the graph displayed in Figure~\ref{fig02} contains $24$ vertices.
An optimal solution covers all the vertices, and its paths are
$u_4$-$u_0$-$v_0$-$v_4$,
$w_0$-$u_1$-$w_1$-$v_1$-$w_2$,
$x_0$-$x_1$-$u_2$-$x_2$-$x_3$,
$x_4$-$x_5$-$v_2$-$x_6$-$x_7$, and
$w_3$-$u_3$-$w_4$-$v_3$-$w_5$.
Assuming the algorithm {\sc Approx1} adds two $4$-paths
$u_0$-$u_1$-$u_2$-$u_3$ and
$v_0$-$v_1$-$v_2$-$v_3$, and then extends them to
$u_0$-$u_1$-$u_2$-$u_3$-$u_4$ and
$v_0$-$v_1$-$v_2$-$v_3$-$v_4$, respectively.
Then none of the three operations can be applied to improve the solution, which covers a total of $10$ vertices only.
This shows that the worst-case performance ratio $2.4$ of the algorithm {\sc Approx1} for \Max{4} is tight.
\end{proof}

\section{A $2$-approximation algorithm for \Max{4}}
One sees that the algorithm {\sc Approx1} is an $O(n^5)$-time tight $2.4$-approximation algorithm for the \Max{4} problem,
which improves the previous best $O(n^5)$-time $4$-approximation algorithm proposed in \cite{KLM22}
and $O(n^7)$-time $(4 - \frac 1{63,700,992} + \epsilon)$-approximation algorithm implied by \cite{Neu21}.
In the amortized analysis for {\sc Approx1},
when a path of $\mcQ - \mcP$ is associated with two vertices of $V(\mcQ) \cap V(\mcP)$,
one half of the vertices on this path is assigned to each of these two vertices.
We will show that when $k = 4$, that is, for the \Max{4} problem,
the assignment can be done slightly better by observing where these two vertices are on the paths of $\mcP$ and then assigning vertices accordingly.
To this purpose, we will need to refine the algorithm using two more local improvement operations,
besides the three Add, Rep and DoubleRep operations.
We denote our algorithm for the \Max{4} problem as {\sc Approx2}.

We again use $\mcP$ to denote the path collection computed by {\sc Approx2}.
Since {\sc Approx2} employs the three Add, Rep and DoubleRep operations, all of which are not applicable at termination,
the structural properties stated in Lemmas~\ref{lemma03}--\ref{lemma05} continue to hold, and we summarize them specifically using $k = 4$.

\begin{lemma}\label{lemma08}
For any path $P \in \mcP$,
\begin{itemize}
\parskip=0pt
\item[(1)]
	$4 \le n(P) \le 7$;
\item[(2)]
	$n(u_j) \le j$ and $n(v_j) \le j$, for any valid index $j$;
\item[(3)]
	if $n(P) = 6$, $n(u_2) > 0$ and $n(v_2) > 0$, then both $u_2$ and $v_2$ are adjacent to a unique vertex in $V - V(\mcP)$,
	and hence they have common extensions;
\item[(4)]
	if $n(P) = 7$, then $n(u_2) = n(v_2) = 0$.
\end{itemize}
\end{lemma}

We also fix $\mcQ$ to denote an optimal collection of vertex-disjoint paths of order at least $4$ that covers the most vertices.
Using Lemma~\ref{lemma08}, we may attach a longest possible extension in $\mcQ - \mcP$ to every vertex of a path of $\mcP$,
giving rise to the {\em worst cases} illustrated in Figure~\ref{fig03}, with respect to the order of the path.
Since a vertex of a path $P \in \mcP$ can be associated with up to two paths in $\mcQ - \mcP$,
the worst-case performance ratio of the algorithm {\sc Approx2} is at most $\max\{\frac {17}7, \frac {14}6, \frac {13}5, \frac 84\} = 2.6$.
Our next two local improvement operations are designed to deal with three of the four worst cases where the path orders are $5, 6$ and $7$.
Afterwards, we will show by an amortization scheme that the average number of vertices assigned to a vertex of $V(\mcP)$ is at most $2$.

\begin{figure}[ht]
\centering
  \setlength{\unitlength}{1bp}%
  \begin{picture}(284.22, 130.41)(0,0)
  \put(0,0){\includegraphics{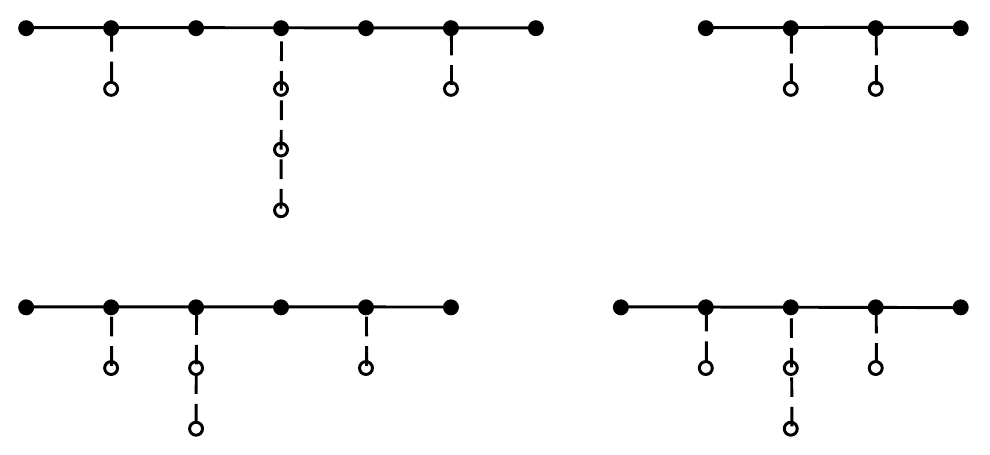}}
  \end{picture}%
\caption{The worst case of a path $P \in \mcP$ with respect to its order,
	where a vertex on the path $P$ is attached with a longest possible extension in $\mcQ - \mcP$, with its edges dashed.
	The edges on the path $P$ are solid and the vertices are shown filled, while the vertices on the extensions are unfilled.\label{fig03}}
\end{figure}

Since the average number of vertices assigned to a vertex on a $4$-path in $\mcP$ is already at most $2$,
the first operation is employed to construct more $4$-paths in $V(\mcP)$, whenever possible.

\begin{operation}\label{op4}
For any two paths $P$ and $P'$ in $\mcP$ of order at least $5$ such that
their vertices are covered exactly by a set of paths in $G$ of order at least $4$ and of which at least one path has order $4$,
the {\em Re-cover($P, P'$)} operation replaces $P$ and $P'$ by this set of paths.
\end{operation}

In other words, the Re-cover($P, P'$) operation removes the two paths $P$ and $P'$ from $\mcP$,
and then uses the set of paths to re-cover the same vertices.
Since there are $O(n^2)$ possible pairs of paths $P$ and $P'$ in $\mcP$,
and by $|V(P) \cup V(P')| \le 14$ the existence of a set of paths re-covering $V(P) \cup V(P')$ can be checked in $O(1)$ time,
we conclude that determining whether or not a Re-cover operation is applicable, and if so then applying it,
can be done in $O(n^2)$ time.
Note that such an operation does not change $|V(\mcP)|$, but it increases the number of $4$-paths by at least $1$.
For example, if $n(P) = 5$ and $n(P') = 7$, and we use $u'_j$'s and $v'_j$'s to label the vertices on $P'$,
then a Re-cover operation is applicable when $u_0$ is adjacent to any one of $u'_0, u'_2, v'_2, v'_0$ (resulting in three $4$-paths);
or if $n(P) = n(P') = 6$, then a Re-cover operation is applicable when $u_0$ is adjacent to any one of $u'_0, u'_1, v'_1, v'_0$
(resulting in three $4$-paths).

\begin{operation}\label{op5}
For a path $P \in \mcP$ such that there is an index $t \in \{2, 3\}$ and an extension $e(u_t)$ with $n(e(u_t)) = t$,
\begin{itemize}
\parskip=0pt
\item[(i)]
	if replacing $P$ by the path $e(u_t)$-$u_t$-$\cdots$-$v_1$-$v_0$ enables a Rep operation,
	then the {\em Look-ahead($P$)} operation first replaces $P$ by $e(u_t)$-$u_t$-$\cdots$-$v_1$-$v_0$ and next executes the Rep operation;
\item[(ii)]
	if $n(P) = 6$ and one of $v_0$ and $v_2$ is adjacent to a vertex $w$,
	such that $w$ is on $e(u_2)$ or on $u_0$-$u_1$ or on another path of $\mcP$ but at distance at most $1$ from one end,
	then the {\em Look-ahead($P$)} operation first replaces $P$ by the path $u_0$-$u_1$-$u_2$-$e(u_2)$ and
	next uses $v_0$-$v_1$-$v_2$ as an extension $e(w)$ to execute a Rep operation.
\end{itemize}
\end{operation}

In some sense, the Look-ahead($P$) operation looks one step ahead to see
whether or not using the extension $e(u_t)$ in various ways would help cover more vertices.
Recall that we can rename the vertices on a path of $\mcP$, if necessary,
and thus the above definition of a Look-ahead operation applies to the vertex $v_2$ symmetrically, if there is an extension $e(v_2)$ with $n(e(v_2)) = 2$.
Also, when $u_t$ is the center vertex of the path $P$ (i.e., $n(P) = 5, 7$),
one should also examine replacing $P$ by the path $u_0$-$u_1$-$\cdots$-$u_t$-$e(u_t)$.

When the first case of a Look-ahead operation applies, its internal Rep operation must involve at least two of the $t$ vertices $u_0, u_1, \ldots, u_{t-1}$%
\footnote{First, at least one of these $t$ vertices $u_0, u_1, \ldots, u_{t-1}$, say $u$, should be on the extension used in the internal Rep operation;
then the extension extends to the/a vertex adjacent to $u$.}
because otherwise an Add or a Rep operation would be applicable before this Look-ahead operation.
Since there are $O(n^3)$ possible extensions at the vertex $u_t$,
it follows from Lemma~\ref{lemma02} that determining whether or not the first case of a Look-ahead operation is applicable, and if so then applying it,
can be done in $O(n^6)$ time.
Note that such an operation increases $|V(\mcP)|$ by at least $1$.

When the second case of a Look-ahead operation applies, we see that after the replacement (which reduces $|V(\mcP)|$ by $1$),
the vertex $w$ is always on a path of $\mcP$ at distance at most $1$ from one end;
and thus the succeeding Rep operation increases $|V(\mcP)|$ by at least $2$.
The net effect is that such a Look-ahead operation increases $|V(\mcP)|$ by at least $1$.
Since there are $O(n^2)$ possible extensions at the vertex $u_2$,
determining whether or not the second case of a Look-ahead operation is applicable, and if so then applying it,
can be done in $O(n^4)$ time.

We summarize the above observations on the two new operations into the following lemma.

\begin{lemma}\label{lemma09}
Given a collection $\mcP$ of vertex-disjoint paths of order in between $4$ and $7$,
determining whether or not one of the two operations {\em Re-cover} and {\em Look-ahead} is applicable, and if so then applying it,
can be done in $O(n^6)$ time.
Each operation either increases $|V(\mcP)|$ by at least $1$,
or keeps $|V(\mcP)|$ unchanged and increases the number of $4$-paths by at least $1$.
\end{lemma}

We are now ready to present the algorithm {\sc Approx2} for the \Max{4} problem, which in fact is very similar to {\sc Approx1}.
It starts with the empty collection $\mcP = \emptyset$;
in each iteration, it determines in order whether any one of the five operations Add, Rep, DoubleRep, Re-cover, and Look-ahead is applicable,
and if so then it applies the operation to update $\mcP$.
During the entire process, $\mcP$ is maintained to be a collection of vertex-disjoint paths of order in between $4$ and $7$.
The algorithm terminates if none of the five operations is applicable for the current $\mcP$, and returns it as the solution.
A simple high level description of the algorithm {\sc Approx2} is depicted in Figure~\ref{fig04}.
From Lemmas~\ref{lemma02} and \ref{lemma09}, the overall running time of {\sc Approx2} is in $O(n^8)$.
Next we show that the worst-case performance ratio of {\sc Approx2} is at most $2$, and thus its higher running time is paid off.
The performance analysis is done through a similar but more careful amortization scheme.

\begin{figure}[htb]
\begin{center}
\framebox{
\begin{minipage}{5in}
Algorithm {\sc Approx2}:\\
Input: A graph $G = (V, E)$;
\begin{itemize}
\parskip=0pt
\item[1.]
	initialize $\mcP = \emptyset$;
\item[2.]
	{\bf while} (one of Add, Rep, DoubleRep, Re-cover and Look-ahead is applicable)
	\begin{itemize}
	\parskip=0pt
	\item[2.1]
		apply the operation to update $\mcP$;
	\end{itemize}
\item[3.]
	return the final $\mcP$.
\end{itemize}
\end{minipage}}
\end{center}
\caption{A high level description of the algorithm {\sc Approx2}.\label{fig04}}
\end{figure}

The amortization scheme assigns the vertices of $V(\mcQ)$ to the vertices of $V(\mcP) \cap V(\mcQ)$.
Firstly, each vertex of $V(\mcQ) \cap V(\mcP)$ is assigned to itself.
Next, recall that $\mcQ - \mcP$ is the collection of sub-paths of the paths of $\mcQ$ after removing those vertices in $V(\mcQ) \cap V(\mcP)$.
By Lemma~\ref{lemma03}, each path $S$ of $\mcQ - \mcP$ is associated with one or two vertices in $V(\mcQ) \cap V(\mcP)$.

When the path $S$ of $\mcQ - \mcP$ is associated with two vertices $v$ and $v'$ of $V(\mcQ) \cap V(\mcP)$,
a half of all the vertices on $S$ are assigned to each of $v$ and $v'$,
except the {\em first} special case below.
In this {\em first} special case,
$n(S) = 1$, $v = u_1$ (or $v = v_1$, respectively) and $v' = u'_3$ on some paths $P, P' \in \mcP$ with $n(P) \ge 5$ and $n(P') = 7$, respectively;
then the whole vertex on $S$ is assigned to $u'_3$ (that is, none is assigned to $u_1$, or $v_1$, respectively).

When the path $S$ of $\mcQ - \mcP$ is associated with only one vertex $v$ of $V(\mcQ) \cap V(\mcP)$,
all the vertices on $S$ are assigned to the vertex $v$,
except the {\em second} and the {\em third} special cases below
where $n(S) = 1$ and $v = u_1$ (or $v = v_1$, respectively) on a path $P \in \mcP$ with $n(P) \ge 5$.
In the {\em second} special case, $S$-$u_1$-[$r$]-$u'_3$ is a subpath of some path $Q \in \mcQ$,
where $u'_3$ is the center vertex of some $7$-path $P' \in \mcP$,
and [$r$] means the vertex $r$ might not exist but if it exists then it is not the center vertex of any $7$-path in $\mcP$;
in this case, a half of the vertex on $S$ is assigned to each of $u_1$ and $u'_3$.%
\footnote{In the second special case, if $r$ exists and $r \notin V(\mcP)$, then it falls into the first special case and, as a result,
	the vertex $u_1$ receives only $\frac 12$ vertex and the vertex $u'_3$ receives $\frac 32$ vertices.}
In the {\em third} special case, $S$-$u_1$-$u'_2$ is a subpath of some path $Q \in \mcQ$,
where $u'_2 \ne u_2$ is on some path $P' \in \mcP$ of order $5$ or $6$;
in this case, a half of the vertex on $S$ is assigned to each of $u_1$ and $u'_2$.

One sees that in the amortization scheme,
all the vertices of $V(\mcQ)$ are assigned to the vertices of $V(\mcP) \cap V(\mcQ)$;
conversely, each vertex of $V(\mcP) \cap V(\mcQ)$ receives itself, plus some or all the vertices on its associated paths of $\mcQ - \mcP$,
and some $u_2$, $v_2$ vertices on $5$-/$6$-paths and some $u_3$ vertices on $7$-paths
could receive some additional vertices not on their associated paths, through the three special cases.
(We remark that a vertex and its received vertices are on the same path in $\mcQ$, and that the vertices of $V(\mcP) - V(\mcQ)$ receive nothing.)

The following four lemmas present the joint effect of the above amortization scheme.

\begin{lemma}\label{lemma10}
For any path $P \in \mcP$, each of the vertices $u_0$ and $v_0$ receives at most $1$ vertex;
each of the vertices $u_1$ and $v_1$ receives at most $\frac 52$ vertices;
when $n(P) = 7$, each of the vertices $u_2$ and $v_2$ receives at most $1$ vertex.
\end{lemma}
\begin{proof}
From Lemma~\ref{lemma08},
$n(u_0) = n(v_0) = 0$, and if $n(P) = 7$ then $n(u_2) = n(v_2) = 0$ too.
Lemma~\ref{lemma01} tells that none of them is associated with any path of $\mcQ - \mcP$.
That is, if any one of them is in $V(\mcQ)$, then it is assigned with itself only, or otherwise it is assigned with nothing.

Also from Lemma~\ref{lemma08}, $n(u_1) \le 1$ and $n(v_1) \le 1$.
We prove the lemma for $u_1$ below.
If $u_1$ is associated with at most one path of $\mcQ - \mcP$, then by the amortization scheme it receives at most $2$ vertices.
If $u_1$ is associated with two paths $\mcQ - \mcP$ denoted as $S_1$ and $S_2$,
then by $n(S_1) = n(S_2) = 1$ we know that $S_1$-$u_1$-$S_2$ is a $3$-subpath of some path $Q \in \mcQ$.
We conclude that at least one of $S_1$ and $S_2$, say $S_1$, is associated with another distinct vertex $v$ of $V(\mcQ) \cap V(\mcP)$.
By the amortization scheme, at most a half of the vertex on $S_1$ is assigned to $u_1$.
Therefore, $u_1$ receives at most $\frac 52$ vertices.
\end{proof}

\begin{lemma}\label{lemma11}
For any $5$-path $P \in \mcP$, if the vertex $u_2$ receives more than $3$ vertices,
then each of the vertices $u_1$ and $v_1$ receives at most $\frac 32$ vertices.
\end{lemma}
\begin{proof}
We prove the lemma for $u_1$ below (for $v_1$, symmetrically).

Assume $u_2$ is on a path $Q \in \mcQ$.
Firstly, if $n(u_2) \le 1$, then from either side of $Q$, $u_2$ receives at most one vertex on the path $Q$.
That is, $u_2$ receives at most three vertices, which contradicts the lemma premise.
We therefore conclude that $u_2$ is associated with a path $S$ of $\mcQ - \mcP$ of order $n(S) = 2$.
Note that the termination condition of the algorithm {\sc Approx2} says that using $S$ as an extension at $u_2$,
no Look-ahead($P$) operation is applicable.

This lemma holds obviously if $u_1$ is not associated with any path of $\mcQ - \mcP$.

If $u_1$ is associated with only one path of $\mcQ - \mcP$, denoted as $S_1$, then by Lemma~\ref{lemma01} $n(S_1) = 1$.
If $S_1$ is associated with another distinct vertex of $V(\mcQ) \cap V(\mcP)$, then at most a half of the vertex on $S_1$ is assigned to $u_1$,
and the lemma is proved.
In the other case, $u_1$-$u_2$ is not an edge in $E(\mcQ)$,
since otherwise the third special case does not apply and hence $u_2$ would receive at most $3$ vertices, contradicting to the lemma premise;
then at least one of the two vertices $u_0$ and $u_1$ is adjacent to a vertex $u$ on some path in $\mcP$ through an edge of $E(Q)$.
Recall that when the algorithm {\sc Approx2} explores a possible Look-ahead operation using $S$ as an extension at $u_2$,
the path $S_1$-$u_1$-$u_0$ is an available extension after the replacement.
We conclude from Lemma~\ref{lemma08} that $u_0$ can be adjacent to only $u'_3$ of some 7-path $P' \in \mcP$,
and $u_1$ can be adjacent to only one of the $u'_2$ and $v'_2$ of some 5-/6-path $P' \in \mcP$, or $u'_3$ of some 7-path $P' \in \mcP$,
due to no applicable Re-cover or Look-ahead operation.
Any one of the above possible scenarios falls into either the second or the third special case in the amortization scheme,
and so that $u_1$ receives only a half of the vertex on $S_1$.
That is, the lemma is proved.

If $u_1$ is associated with two paths $S_1$ and $S_2$ of $\mcQ - \mcP$, that is, $S_1$-$u_1$-$S_2$ is a subpath of a path in $\mcQ$,
then by Lemma~\ref{lemma01} $n(S_1) = n(S_2) = 1$.
Recall that when the algorithm {\sc Approx2} explores a possible Look-ahead operation using $S$ as an extension at $u_2$,
the path $S_1$-$u_1$-$S_2$ is an available extension after the replacement.
We conclude from Lemma~\ref{lemma08} that $S_1$ and $S_2$ can be adjacent to only the center vertex of some $7$-path in $\mcP$,
due to no applicable Look-ahead operation.
Apparently one least one of $S_1$ and $S_2$, say $S_1$, is adjacent to the center vertex $u'_3$ of some $7$-path $P' \in \mcP$.
By the first special case in the amortization scheme, the whole vertex on $S_1$ is assigned to $u'_3$.
If $S_2$ is also adjacent to the center vertex $u''_3$ of some $7$-path $P'' \in \mcP$, then $u_1$ receives only $1$ vertex, which is itself;
otherwise, by the second special case in the amortization scheme, a half of the vertex on $S_2$ is assigned to $u'_3$,
and thus $u_1$ receives a total of $\frac 32$ vertices.
The lemma is proved.
\end{proof}

\begin{lemma}\label{lemma12}
For any $6$-path $P \in \mcP$, if the vertices $u_2$ and $v_2$ together receive more than $5$ vertices,
then the total number of vertices received by all the vertices of $P$ is at most $12$.
\end{lemma}
\begin{proof}
From Lemma~\ref{lemma08}, for any path $S_1$ of $\mcQ - \mcP$ associated with $u_2$ and
any path $S_2 \ne S_1$ of $\mcQ - \mcP$ associated with $v_2$, they are vertex-disjoint and thus $n(S_1) + n(S_2) \le 2$,
and furthermore if $n(S_1) + n(S_2) = 2$ then one of them is null.
This implies that $u_2$ and $v_2$ together can have at most two associated paths in $\mcQ - \mcP$,
and if there are two, then both are associated with either $u_2$ only or $v_2$ only.

If $u_2$ is associated with no path in $\mcQ - \mcP$,
then by the third special case in the amortization scheme $u_2$ can receive at most $2$ vertices.
If $u_2$ is associated with a path $S$ in $\mcQ - \mcP$ with $n(S) = 1$,
then $u_2$ can receive at most $\frac 52$ vertices.
If $u_2$ is associated with two paths $S_1$ and $S_2$ in $\mcQ - \mcP$ with $n(S_1) = n(S_2) = 1$,
then $u_2$ can receive at most $\frac 52$ vertices.
The same argument applies to $v_2$,
and from the lemma premise we conclude that $u_2$ is associated with a path of $\mcQ - \mcP$, denoted as $S_1$, with $n(S_1) = 2$
(or for $v_2$, which can be proved symmetrically).

We claim that $v_2$ receives no vertex other than itself.
This claim is true since otherwise $v_2$ is also associated with $S_1$ or $v_2$ has to be involved in the third special case in the amortization scheme,
that is, $v_2$ would be adjacent to one of the two vertices on $S_1$ or the vertex $u'_1 \ne v_1$ on some path $P' \in \mcP$.
However, this implies that a Look-ahead operation is applicable by
using $S_1$ as the extension at $u_2$ and $v_0$-$v_1$-$v_2$ as the extension for the subsequent Rep operation (see Operation~\ref{op5}).

Given that $u_2$ receives at most $5$ vertices,
the above claim on $v_2$ tells that $u_2$ and $v_2$ together receive at most $6$ vertices.
By Lemma~\ref{lemma10}, if one of $u_0$ and $v_0$ is not in $V(\mcQ)$, or one of $u_1$ and $v_1$ is not associated with any path of $\mcQ - \mcP$,
then the lemma is proved.
Below we consider the remaining case, on $u_1$ specifically to show that it receives at most $\frac 32$ vertices.

If $u_1$ is associated with only one path of $\mcQ - \mcP$, denoted as $S_3$, then by Lemma~\ref{lemma01} $n(S_3) = 1$.
If $S_3$ is associated with another distinct vertex of $V(\mcQ) \cap V(\mcP)$, then at most a half of the vertex on $S_3$ is assigned to $u_1$,
and so $u_1$ receives at most $\frac 32$ vertices.
In the other case, $u_0$ or $u_1$ is adjacent to a vertex $u \ne u_1$ on some path in $\mcP$ through an edge of $E(\mcQ)$,
and recall that when the algorithm {\sc Approx2} explores a possible Look-ahead operation using $S_1$ as an extension at $u_2$,
the path $S_3$-$u_1$-$u_0$ is an available extension after the replacement.
We conclude from Lemma~\ref{lemma08} that $u_0$ can be adjacent to only $u'_3$ of some 7-path $P' \in \mcP$,
and $u_1$ can be adjacent to only one of the $u'_2$ and $v'_2$ of some 5-/6-path $P' \in \mcP$, or $u'_3$ of some 7-path $P' \in \mcP$,
due to no applicable Re-cover or Look-ahead operation.
Any one of the above scenarios falls into the second or the third special case in the amortization scheme,
and so that $u_1$ receives only a half of the vertex on $S_3$.
That is, $u_1$ receives a total of $\frac 32$ vertices.

If $u_1$ is associated with two paths $S_3$ and $S_4$ of $\mcQ - \mcP$, that is, $S_3$-$u_1$-$S_4$ is a subpath of a path in $\mcQ$,
then by Lemma~\ref{lemma01} $n(S_3) = n(S_4) = 1$.
Recall that when the algorithm {\sc Approx2} explores a possible Look-ahead operation using $S_1$ as an extension at $u_2$,
the path $S_3$-$u_1$-$S_4$ is an available extension after the replacement.
We conclude from Lemma~\ref{lemma08} that $S_3$ and $S_4$ can be adjacent to only the center vertex of some $7$-path in $\mcP$,
due to no applicable Look-ahead operation.
Apparently one least one of $S_3$ and $S_4$, say $S_3$, is adjacent to the center vertex $u'_3$ of some $7$-path $P' \in \mcP$.
By the first special case in the amortization scheme, the whole vertex on $S_3$ is assigned to $u'_3$.
If $S_4$ is also adjacent to the center vertex $u''_3$ of some $7$-path $P'' \in \mcP$, then $u_1$ receives only $1$ vertex, which is itself;
otherwise, by the second special case in the amortization scheme, a half of the vertex on $S_4$ is assigned to $u'_3$,
and thus $u_1$ receives a total of $\frac 32$ vertices.

In summary, for the remaining case,
$u_2$ and $v_2$ together receive at most $6$ vertices, and $u_1$ receives a total of at most $\frac 32$ vertices.
By Lemma~\ref{lemma10}, the total number of vertices received by all the vertices of $P$ is at most $12$.
This finishes the proof of the lemma.
\end{proof}

\begin{lemma}\label{lemma13}
For any $7$-path $P \in \mcP$, if the vertex $u_3$ receives more than $5$ vertices,
then the vertices $u_0$ and $u_1$ ($v_0$ and $v_1$, respectively) together receive at most $\frac 52$ vertices.
\end{lemma}
\begin{proof}
Let $S_1$ be a path of $\mcQ - \mcP$ associated with the vertex $u_3$, that is, $S_1$-$u_3$ is a subpath of some path $Q \in \mcQ$.
If $n(S_1) \le 2$, then the number of vertices on this side of the path $Q$ assigned to $u_3$ is at most $2$.
Since $u_3$ receives more than $5$ vertices, and there are at most two paths of $\mcQ - \mcP$ associated with the vertex $u_3$,
we conclude that, at least one associated path has order $3$, and we assume without loss of generality that $n(S_1) = 3$.
Furthermore, there are two paths of $\mcQ - \mcP$ associated with $u_3$.

If the vertex $u_1$ is not associated with any path of $\mcQ - \mcP$, or if $u_0 \notin V(\mcQ)$,
then by Lemma~\ref{lemma10} the vertices $u_0$ and $u_1$ together receive at most $\frac 52$ vertices,
and the lemma is proved.
Below we consider the remaining case, on $u_1$ specifically to show that it receives at most $\frac 32$ vertices.

If $u_1$ is associated with only one path of $\mcQ - \mcP$, denoted as $S_3$, then by Lemma~\ref{lemma01} $n(S_3) = 1$.
If $S_3$ is associated with another distinct vertex of $V(\mcQ) \cap V(\mcP)$, then at most a half of the vertex on $S_3$ is assigned to $u_1$,
and so $u_1$ receives at most $\frac 32$ vertices.
In the other case, at least one of $u_0, u_1, u_2$ is adjacent to a vertex $u \ne u_1$ on some path in $\mcP$ through an edge of $E(\mcQ)$,
and recall that when the algorithm {\sc Approx2} explores a possible Look-ahead operation using $S_1$ as an extension at $u_3$,
the path $u_0$-$u_1$-$u_2$ is an available extension after the replacement.
We conclude from Lemma~\ref{lemma08} that $u_0$ and $u_2$ can be adjacent to only $u'_3$ of some 7-path $P' \in \mcP$,
and $u_1$ can be adjacent to only one of the $u'_2$ and $v'_2$ of some 5-/6-path $P' \in \mcP$, or $u'_3$ of some 7-path $P' \in \mcP$,
due to no applicable Re-cover or Look-ahead operation.
Any one of the above scenarios falls into the second or the third special case in the amortization scheme,
and so that $u_1$ receives only a half of the vertex on $S_3$.
That is, $u_1$ receives a total of $\frac 32$ vertices.

If $u_1$ is associated with two paths $S_3$ and $S_4$ of $\mcQ - \mcP$, that is, $S_3$-$u_1$-$S_4$ is a subpath of a path in $\mcQ$,
then by Lemma~\ref{lemma01} $n(S_3) = n(S_4) = 1$.
Recall that when the algorithm {\sc Approx2} explores a possible Look-ahead operation using $S_1$ as an extension at $u_3$,
the path $S_3$-$u_1$-$S_4$ is an available extension after the replacement.
We conclude from Lemma~\ref{lemma08} that $S_3$ and $S_4$ can be adjacent to only the center vertex of some $7$-path in $\mcP$,
due to no applicable Look-ahead operation.
Apparently one least one of $S_3$ and $S_4$, say $S_3$, is adjacent to the center vertex $u'_3$ of some $7$-path $P' \in \mcP$.
By the first special case in the amortization scheme, the whole vertex on $S_3$ is assigned to $u'_3$.
If $S_4$ is also adjacent to the center vertex $u''_3$ of some $7$-path $P'' \in \mcP$, then $u_1$ receives only $1$ vertex, which is itself;
otherwise, by the second special case in the amortization scheme, a half of the vertex on $S_4$ is assigned to $u'_3$,
and thus $u_1$ receives a total of $\frac 32$ vertices.

In summary, for the remaining case, $u_0$ and $u_1$ together receive at most $\frac 52$ vertices.
This finishes the proof of the lemma.
\end{proof}

\begin{theorem}\label{thm02}
The algorithm {\sc Approx2} is an $O(|V|^8)$-time $2$-approximation algorithm for the \Max{4} problem,
and $\frac {16}9$ is a lower bound on its performance ratio.
\end{theorem}
\begin{proof}
Similar to the proof of Theorem~\ref{thm01}, we will show that $\frac {r(P)}{n(P)} \le 2$ for any path $P \in \mcP$,
where $r(P)$ denotes the total number of vertices received by all the vertices on $P$ through the amortization scheme.
We do this by differentiating the path order $n(P)$, which is in between $4$ and $7$.

Case 1. $n(P) = 4$. 
By Lemma~\ref{lemma10}, we have $r(P) \le 1 + \frac 52 + \frac 52 + 1 = 7$.
It follows that $\frac {r(P)}{n(P)} \le \frac 74$.

Case 2. $n(P) = 5$.
If the vertex $u_2$ receives at most $3$ vertices, then by Lemma~\ref{lemma10}, we have $r(P) \le 3 + 7 = 10$.
Otherwise, $u_2$ receives at most $5$ vertices, and by Lemma~\ref{lemma11} each of $u_1$ and $v_1$ receives at most $\frac 32$ vertices.
It follows again by Lemma~\ref{lemma10} that $r(P) \le 5 + 5 = 10$.
That is, either way we have $r(P) \le 10$ and thus $\frac {r(P)}{n(P)} \le 2$.

Case 3. $n(P) = 6$. 
If the vertices $u_2$ and $v_2$ together receive at most $5$ vertices, then by Lemma~\ref{lemma10}, we have $r(P) \le 5 + 7 = 12$.
Otherwise, by Lemma~\ref{lemma12} we have $r(P) \le 12$.
That is, either way we have $r(P) \le 12$ and thus $\frac {r(P)}{n(P)} \le 2$.

Case 4. $n(P) = 7$. 
If the vertex $u_3$ receives at most $5$ vertices, then by Lemma~\ref{lemma10}, we have $r(P) \le 5 + 2 + 7 = 14$.
Otherwise, $u_3$ receives at most $7$ vertices,
and by Lemma~\ref{lemma13} the vertices $u_0$ and $u_1$ ($v_0$ and $v_1$, respectively) together receive at most $\frac 52$ vertices.
It follows again by Lemma~\ref{lemma10} that $r(P) \le 7 + 2 + 5 = 14$.
That is, either way we have $r(P) \le 14$ and thus $\frac {r(P)}{n(P)} \le 2$.

This proves that {\sc Approx2} is a $2$-approximation algorithm.

One might wonder whether the performance analysis can be done better.
Though we are not able to show the tightness of the performance ratio $2$,
we give below a graph to show that $\frac {16}9$ is a lower bound.

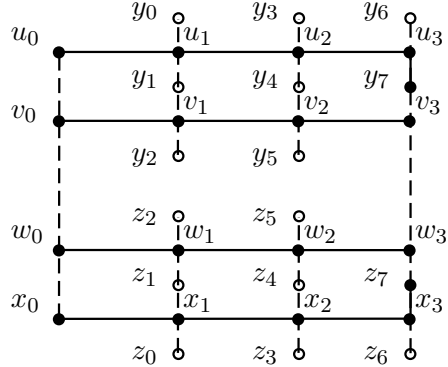
\begin{figure}[ht]
\centering
  \setlength{\unitlength}{1bp}%
  \begin{picture}(194.55, 152.67)(0,0)
  \put(0,0){\includegraphics{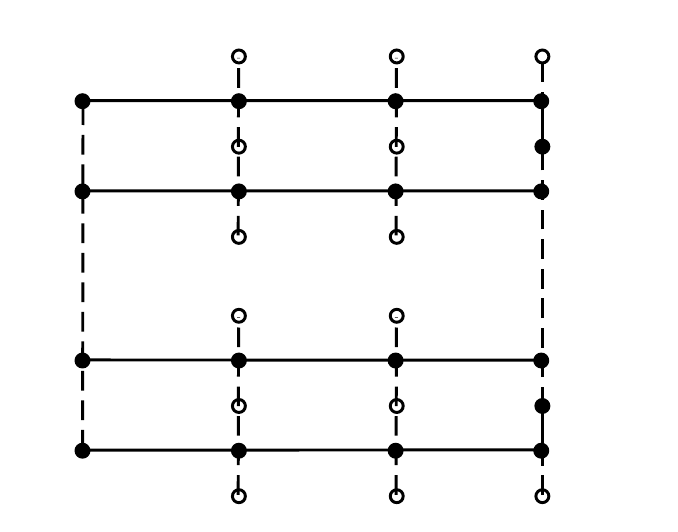}}
  \put(5.67,128.37){\fontsize{11.38}{13.66}\selectfont $u_0$}
  \put(70.61,128.37){\fontsize{11.38}{13.66}\selectfont $u_1$}
  \put(116.07,128.37){\fontsize{11.38}{13.66}\selectfont $u_2$}
  \put(5.67,99.15){\fontsize{11.38}{13.66}\selectfont $v_0$}
  \put(70.61,102.40){\fontsize{11.38}{13.66}\selectfont $v_1$}
  \put(116.07,102.40){\fontsize{11.38}{13.66}\selectfont $v_2$}
  \put(5.67,53.69){\fontsize{11.38}{13.66}\selectfont $w_0$}
  \put(70.61,53.69){\fontsize{11.38}{13.66}\selectfont $w_1$}
  \put(116.07,53.69){\fontsize{11.38}{13.66}\selectfont $w_2$}
  \put(5.67,27.72){\fontsize{11.38}{13.66}\selectfont $x_0$}
  \put(70.61,27.72){\fontsize{11.38}{13.66}\selectfont $x_1$}
  \put(116.07,27.72){\fontsize{11.38}{13.66}\selectfont $x_2$}
  \put(51.13,138.11){\fontsize{11.38}{13.66}\selectfont $y_0$}
  \put(51.13,112.14){\fontsize{11.38}{13.66}\selectfont $y_1$}
  \put(51.13,82.92){\fontsize{11.38}{13.66}\selectfont $y_2$}
  \put(96.59,138.11){\fontsize{11.38}{13.66}\selectfont $y_3$}
  \put(96.59,112.14){\fontsize{11.38}{13.66}\selectfont $y_4$}
  \put(96.59,82.92){\fontsize{11.38}{13.66}\selectfont $y_5$}
  \put(51.13,8.12){\fontsize{11.38}{13.66}\selectfont $z_0$}
  \put(51.13,37.46){\fontsize{11.38}{13.66}\selectfont $z_1$}
  \put(51.13,60.30){\fontsize{11.38}{13.66}\selectfont $z_2$}
  \put(96.59,8.12){\fontsize{11.38}{13.66}\selectfont $z_3$}
  \put(96.59,37.46){\fontsize{11.38}{13.66}\selectfont $z_4$}
  \put(96.59,60.30){\fontsize{11.38}{13.66}\selectfont $z_5$}
  \put(158.02,128.37){\fontsize{11.38}{13.66}\selectfont $u_3$}
  \put(158.02,102.40){\fontsize{11.38}{13.66}\selectfont $v_3$}
  \put(158.02,53.69){\fontsize{11.38}{13.66}\selectfont $w_3$}
  \put(158.02,27.72){\fontsize{11.38}{13.66}\selectfont $x_3$}
  \put(138.53,138.11){\fontsize{11.38}{13.66}\selectfont $y_6$}
  \put(138.53,112.14){\fontsize{11.38}{13.66}\selectfont $y_7$}
  \put(138.53,8.12){\fontsize{11.38}{13.66}\selectfont $z_6$}
  \put(138.53,37.46){\fontsize{11.38}{13.66}\selectfont $z_7$}
  \end{picture}%
\caption{A graph of order $32$ to show that the performance ratio of the algorithm {\sc Approx2} is lower bounded by $\frac {16}9$.
	All the edges in the graph are shown, either solid or dashed.
	The $18$ filled vertices are covered by the collection of two $5$-paths and two $4$-paths computed by {\sc Approx2},
	and the edges on these paths are shown solid;
	the edges on an optimal collection of paths, which covers all the vertices, are shown dashed.\label{fig05}}
\end{figure}

The graph displayed in Figure~\ref{fig05} contains $32$ vertices.
An optimal solution covers all the vertices, and its paths are (vertical)
$u_0$-$v_0$-$w_0$-$x_0$,
$y_0$-$u_1$-$y_1$-$v_1$-$y_2$,
$z_0$-$x_1$-$z_1$-$w_1$-$z_2$,
$y_3$-$u_2$-$y_4$-$v_2$-$y_5$,
$z_3$-$x_2$-$z_4$-$w_2$-$z_5$,
$y_6$-$u_3$-$y_7$-$v_3$, and $w_3$-$z_7$-$x_3$-$z_6$.
Assuming the algorithm {\sc Approx2} adds four $4$-paths (horizontal)
$u_0$-$u_1$-$u_2$-$u_3$,
$v_0$-$v_1$-$v_2$-$v_3$,
$w_0$-$w_1$-$w_2$-$w_3$,
$x_0$-$x_1$-$x_2$-$x_3$, and then extends the first and the last to
$u_0$-$u_1$-$u_2$-$u_3$-$y_7$,
$x_0$-$x_1$-$x_2$-$x_3$-$z_7$.
Then none of the five operations can be applied to improve the solution, which covers a total of $18$ vertices only.
That is, the algorithm {\sc Approx2} only achieves a performance ratio of $\frac {16}9$ on the graph.
\end{proof}

\section{Numerical experiments}
\subsection{Instance graph generation}
Given the NP-hardness of the \Max{k} problem,
our simulation process generates instance graphs for each of which an optimal collection of paths covers all the vertices.
This way, we will not deal with computing an optimal solution.

In details, for a specified order $n$,
the simulation process each time selects a uniformly random integer $\ell$ from the range $[k, k+1, \cdots, 2k-1]$
and creates an $\ell$-path on $\ell$ new vertices.
It repeats so as long as there are at least $3k-1$ vertices to be generated.
When there are at most $2k-1$ vertices to be generated, the process simply generates the last path to make up exactly $n$ vertices in total;
or otherwise it generates the last two paths both to have order in the range $[k, k+1, \cdots, 2k-1]$ to make up exactly $n$ vertices in total.

Afterwards, the simulation process is given a density parameter $d \in [0, 1]$,
and uses it as the probability to add each edge, which is not on the above generated paths, to the graph.

Lastly, the simulation process randomly permutes the vertices.
Such a step is to prevent any algorithm from taking advantage of the vertex order.
At the end, the vertices are indexed from $0$ to $n-1$.

An instance graph generated in this way is labeled with the tuple $(k, n, d, i)$,
and we report the computational results using the ranges for these four parameters
$\{4, 8, 12, 24\}$, $\{50, 100, 200, 400\}$, $[0, 0.02]$ at increment $0.001$ and $\{0.025, 0.03, 0.04\}$, $[0, 99]$ at increment $1$,
respectively.

\subsection{Results}
The above instance simulation process and the two algorithms {\sc Approx1} and {\sc Approx2} are implemented in Python
(Version 3.7.0, with the supporting libraries ``numpy'', ``sys'', ``time'').
Recall that the theoretical time complexities of {\sc Approx1} and {\sc Approx2} are $O(n^{k+1})$ and $O(n^8)$, respectively.
The experiments were done on a shared Linux computer with quad-core 3.0GHz CPU and 16GB Ram.

For each triple $(k, n, d)$, we generated independently $100$ instance graphs,
on which we collected the average and the worst performance ratios of {\sc Approx1}, and of {\sc Approx2} when $k = 4$.
These ratios are plotted in Figures~\ref{fig06} and \ref{fig07}.
We also collected the average running times of the two algorithms and plotted them in Figure~\ref{fig09}.

\begin{figure}[ht]
\centering
\mbox{\hspace{-0.2in}}\includegraphics[width=1.0\linewidth]{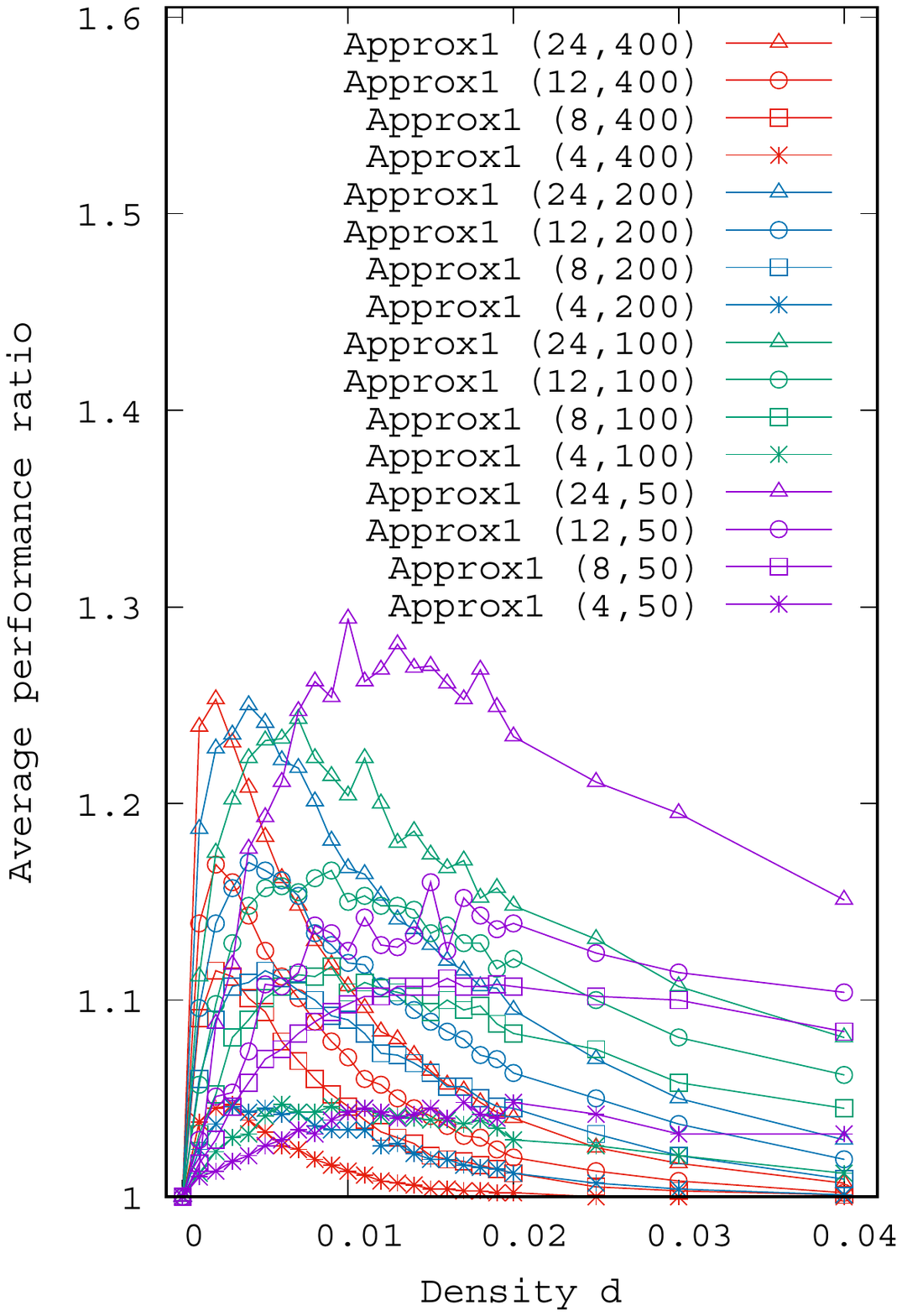}%
\mbox{\hspace{-3.2in}}\includegraphics[width=1.0\linewidth]{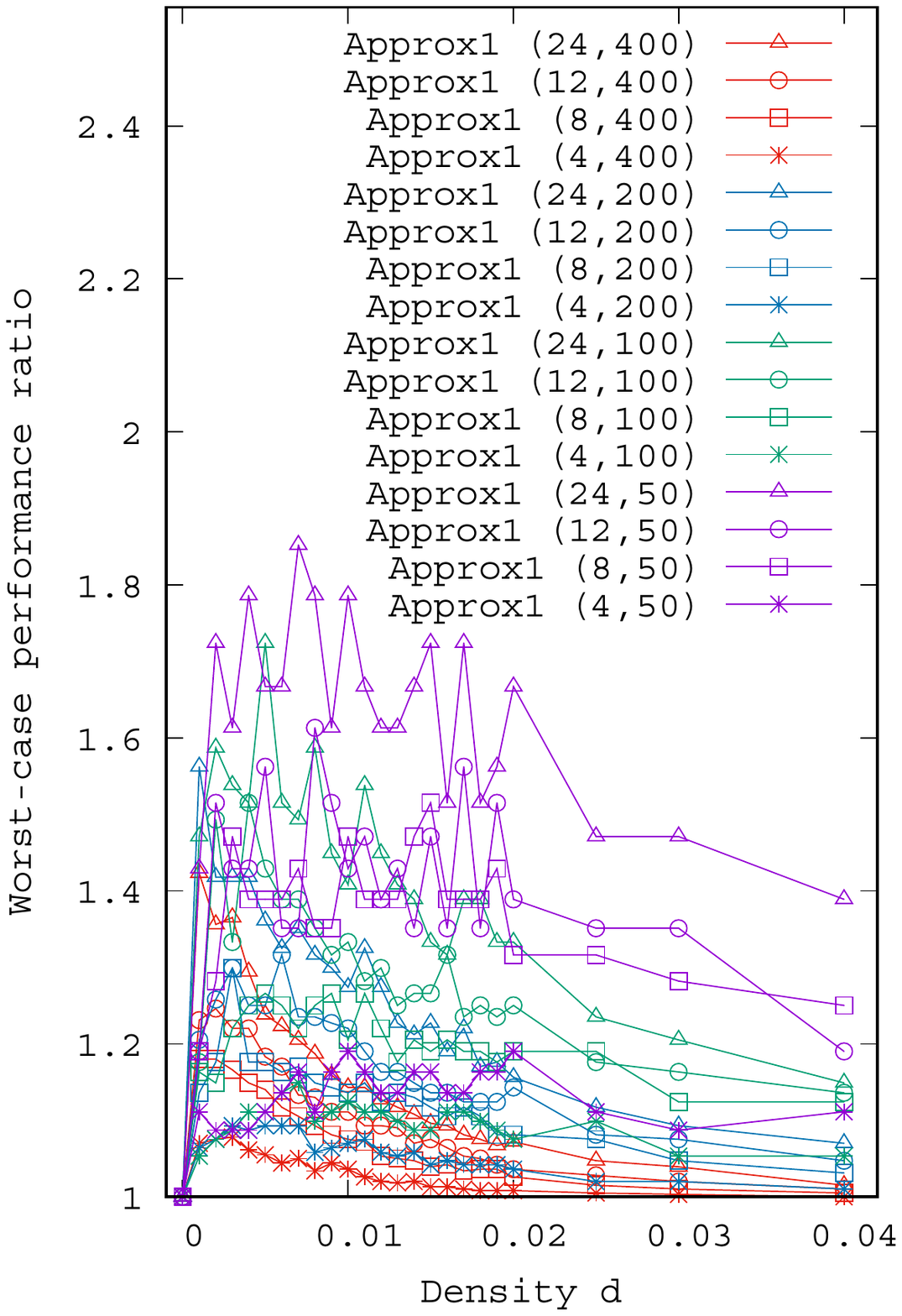}
\caption{The average (left) and the worst (right) performance ratios of {\sc Approx1} for \Max{k},
	on the $100$ instance graphs generated with the triple $(k, n, d)$.\label{fig06}}
\end{figure}

From the collected performance ratios for {\sc Approx1} as plotted in Figure~\ref{fig06},
we observed that for both the average and the worst performance ratios, the peaks shift right along with the edge density $d$.
That is, by the coloring scheme in the figure, from left to right they are red, blue, green and purple, in order.
More specifically, the colored peaks are roughly at $d =$ 0.002-0.003, 0.004-0.005, 0.008-0.01, and 0.01-0.02,
which is around $\frac 1n$ for $n = 400, 200, 100$ and $50$, respectively.
This means, the number of edges in a difficult instance graph is roughly $2n$.
Such an observation is further confirmed on the instance graphs simulated for \Max{4},
on which the collected performance ratios for {\sc Approx1} and for {\sc Approx2}, as plotted in Figure~\ref{fig07}, show the same pattern.

We also observed that when the edge density $d = 0$, that is, no additional edges are added to the simulated optimal path collection,
both {\sc Approx1} and {\sc Approx2} found the optimal solution.
This is non-surprising due to the Rep operation (Operation~\ref{op2}) which extends any subpath to the simulated path.
On the other hand, the average performance ratios drop quickly when $d$ gets larger.
Furthermore, for larger orders $n$, they drop faster (for example, the red plots associated with $n = 400$ in Figures~\ref{fig06} and \ref{fig07}).

Among all the edge density values, the worst (i.e. the peak) of the average performance ratios of {\sc Approx1}
appears to be independent of the order $n$ of the instance graph.
For example, when $k = 24$, they are $1.253, 1.250, 1.243, 1.294$, for $n = 400, 200, 100, 50$, respectively.
However, the worst of the worst performance ratios seems to get better as the order $n$ of the instance graph increases.
For example, when $k = 24$, they are $1.423, 1.562, 1.724, 1.852$, for $n = 400, 200, 100, 50$, respectively.
This observation is further confirmed on the instances simulated for \Max{4}, and in fact even applies to {\sc Approx2},
as one sees their performance rations plotted in Figure~\ref{fig07}.
Specifically, the peaks of the average performance ratios of {\sc Approx1} and {\sc Approx2} are around $1.047$ and $1.027$, respectively.

That said, with respect to $k$, the average performance ratios of {\sc Approx1} gets worse as $k$ increases.
This can be seen by observing the groups of four line plots of the same color in the left side of Figure~\ref{fig06}.
For the worst performance ratios, the tendency remains true in general, but could be zigzag as shown in the right side of Figure~\ref{fig06}.

\begin{figure}[ht]
\centering
\mbox{\hspace{-0.2in}}\includegraphics[width=1.0\linewidth]{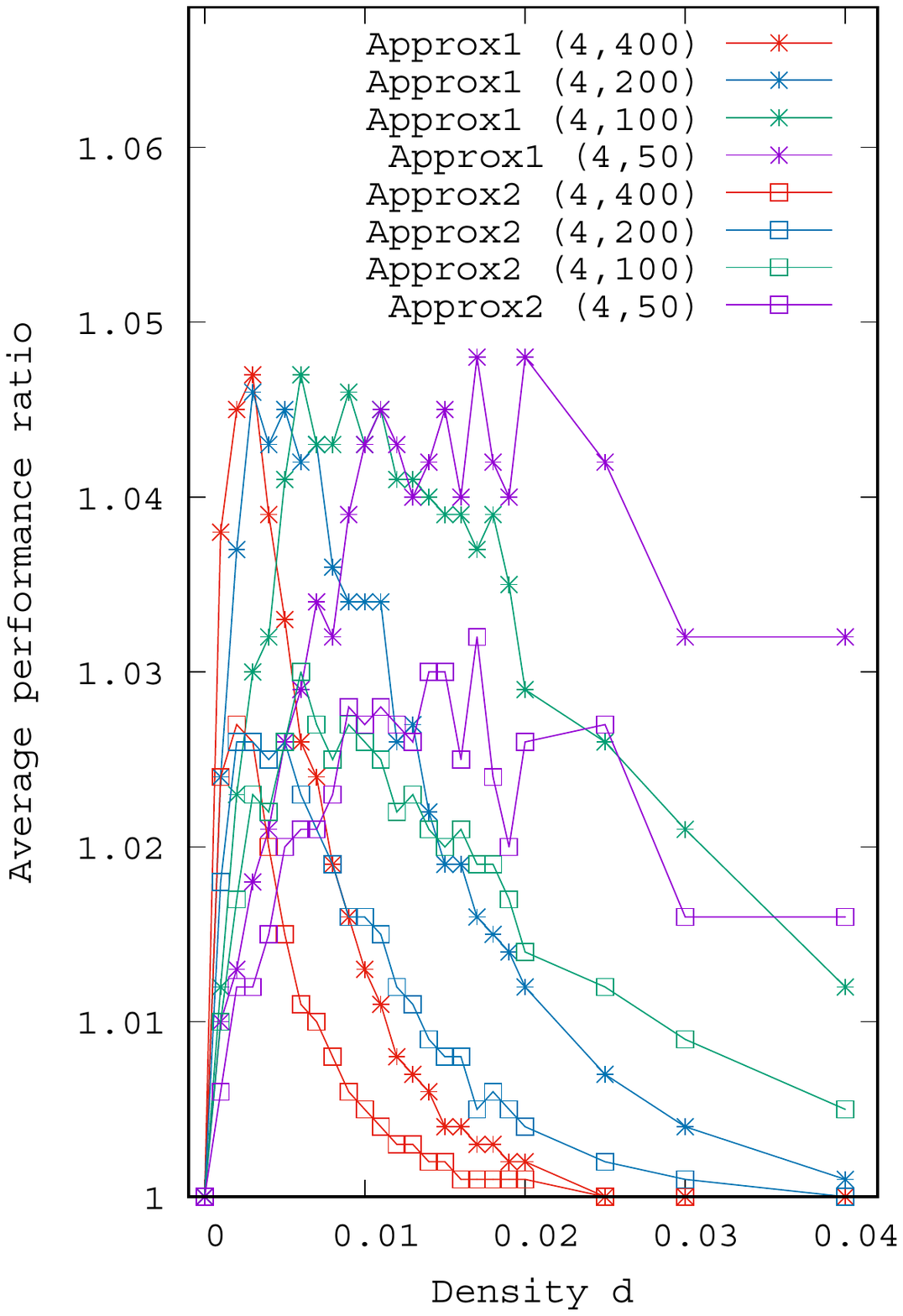}%
\mbox{\hspace{-3.2in}}\includegraphics[width=1.0\linewidth]{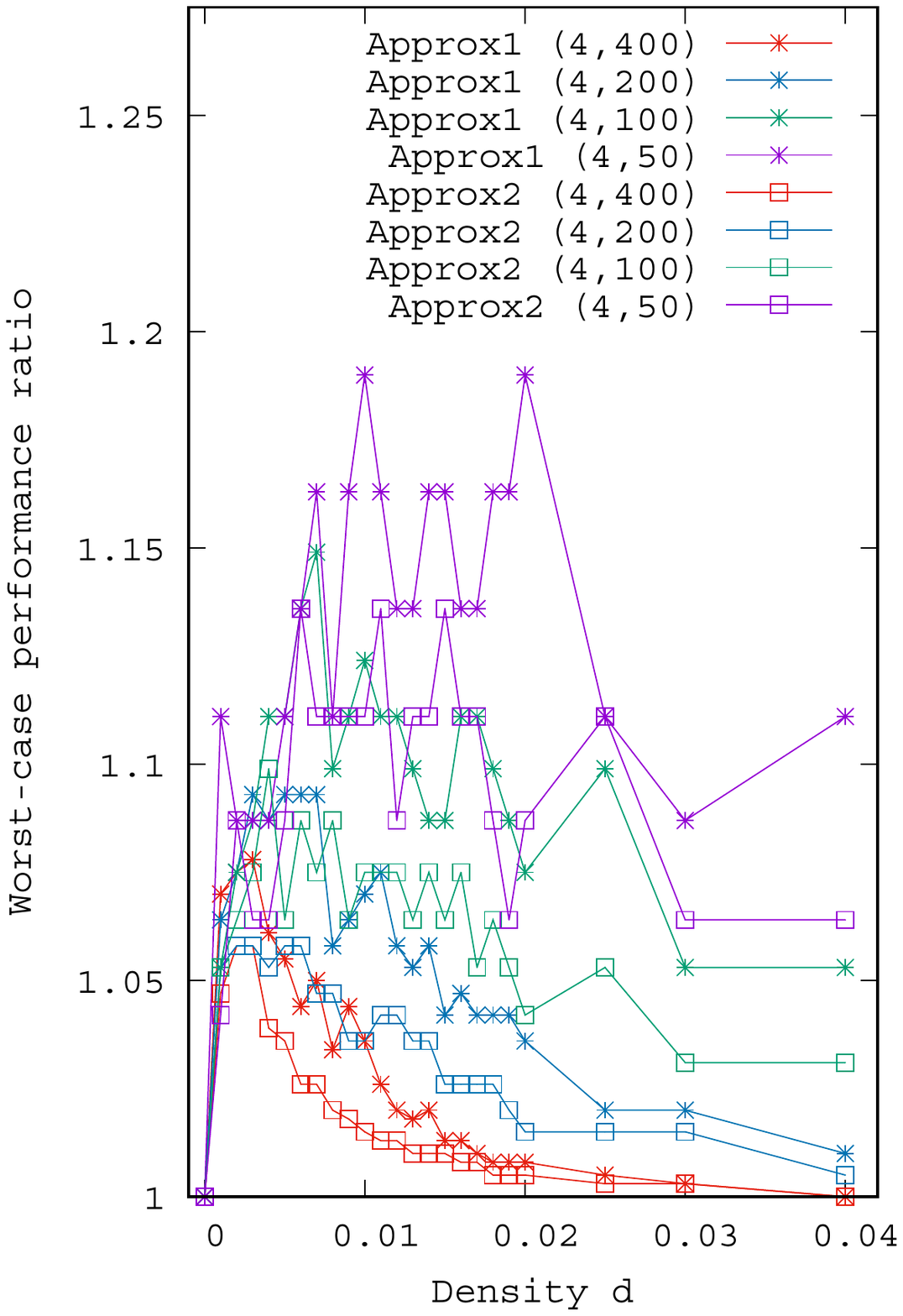}
\caption{The average and the worst performance ratios of {\sc Approx1} and of {\sc Approx2} for \Max{4},
	on the $100$ instance graphs generated with the triple $(4, n, d)$.\label{fig07}}
\end{figure}

On all instance graphs simulated for \Max{4}, {\sc Approx2} consistently outperformed {\sc Approx1}
in terms of both the average and the worst performance ratios (of course, except those instances where {\sc Approx1} achieves the optimal solutions).
It is interesting to observe that the patterns of the average performance ratios of {\sc Approx1} and {\sc Approx2}
look roughly the same as shown in the left side of Figure~\ref{fig07};
their peaks are also at roughly the same edge density $d$, and the largest difference between them appears to be at their peaks.
The worst performance ratios of the two algorithms show the similar tendency (right side of Figure~\ref{fig07}),
but not as clearly as the average performance ratios.
As we have said earlier, in terms of the difference between the average performance ratios of the two algorithms,
we observed that the order of the instance graph does not seem to have a significant impact;
this is further confirmed by the line plots of asterisks in Figure~\ref{fig08},
which are the average difference between the performance ratios on each set of $100$ instances.
However, the line plots of squares in Figure~\ref{fig08} show that the largest difference between the performance ratios of the two algorithms
decreases along with the the order of the instance graph.

\begin{figure}[ht]
\centering
	\includegraphics[width=0.9\linewidth]{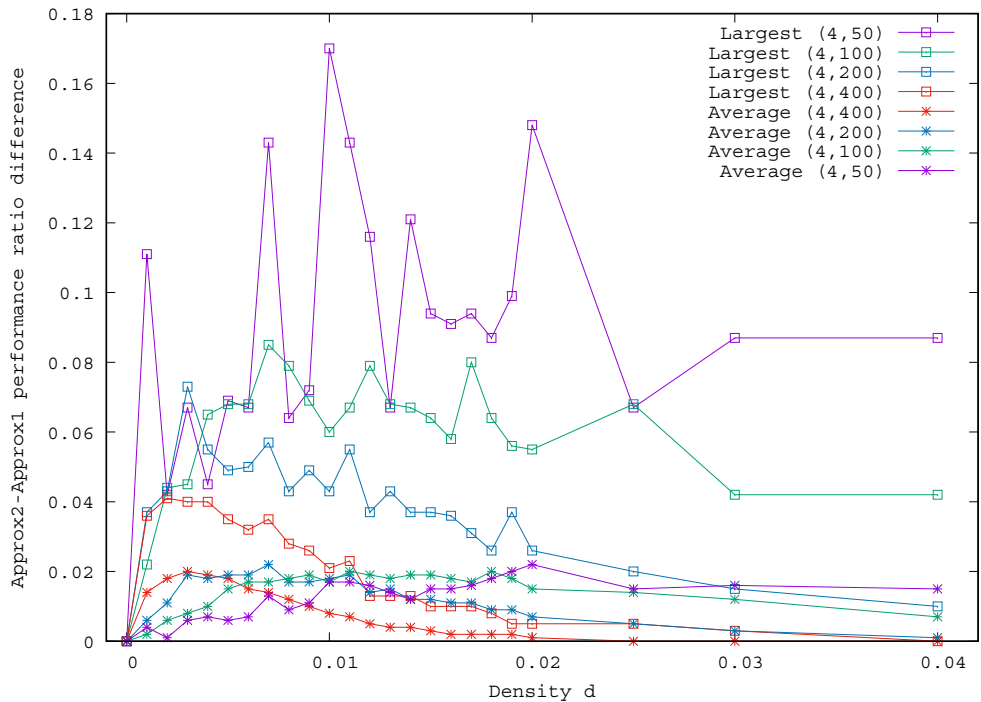}
\caption{The average and the largest difference between the performance ratios of {\sc Approx1} and of {\sc Approx2} for \Max{4},
	on the $100$ instance graphs generated with the triple $(4, n, d)$.\label{fig08}}
\end{figure}

Note that for \Max{k}, the basic operation in the algorithms {\sc Approx1} and {\sc Approx2} is
examining whether or not an edge is inside the instance graph.
In our experiments, this is to check whether an adjacency matrix entry is $1$ or not.
That is, the basic operation is not really an arithmetic operation, and thus we chose to implement our algorithms in Python.
We collected all the average running times of {\sc Approx1}, and of {\sc Approx2} when $k = 4$,
over the $100$ instance graphs associated with the same triple $(k, n, d)$,
and plotted them in Figure~\ref{fig09}.
As expected, for larger graphs, both algorithms took longer times,
shown in Figure~\ref{fig09} by the coloring scheme,
where red lines correspond to instance graphs of order $n = 400$ and purple lines correspond to order $n = 50$.
One sees that the choice of Python is not bad, as the largest average running time in our experiments is slightly less than $20$ seconds.

Observing the line plots of the same shape in the left side of Figure~\ref{fig09}, corresponding to a common value of $k$,
we can safely claim that the average running time increases along with $k$, and it increases more quickly than with $n$.
For the pair $(n, k)$, the average running time shows a similar pattern as the performance ratio with respect to the edge density $d$,
that each line plot has a peak although most peaks are at $d = 0.001$.
It is worth pointing out that at $d = 0$, {\sc Approx1} spent a non-trivial amount of time on instance graphs of order $n = 200$ and $n = 400$
(left side of Figure~\ref{fig09}).
Between {\sc Approx1} and {\sc Approx2}, we observed that on average {\sc Approx2} took much longer on the same instance (right side of Figure~\ref{fig09},
by comparing the line plots of the same color).
Nevertheless, the longer running time is paid off by its performance ratio,
which is roughly $44\%$ better than {\sc Approx1} (left side of Figure~\ref{fig07} and the line plots with the asterisks in Figure~\ref{fig08}).

\begin{figure}[ht]
\centering
\mbox{\hspace{-0.2in}}\includegraphics[width=1.0\linewidth]{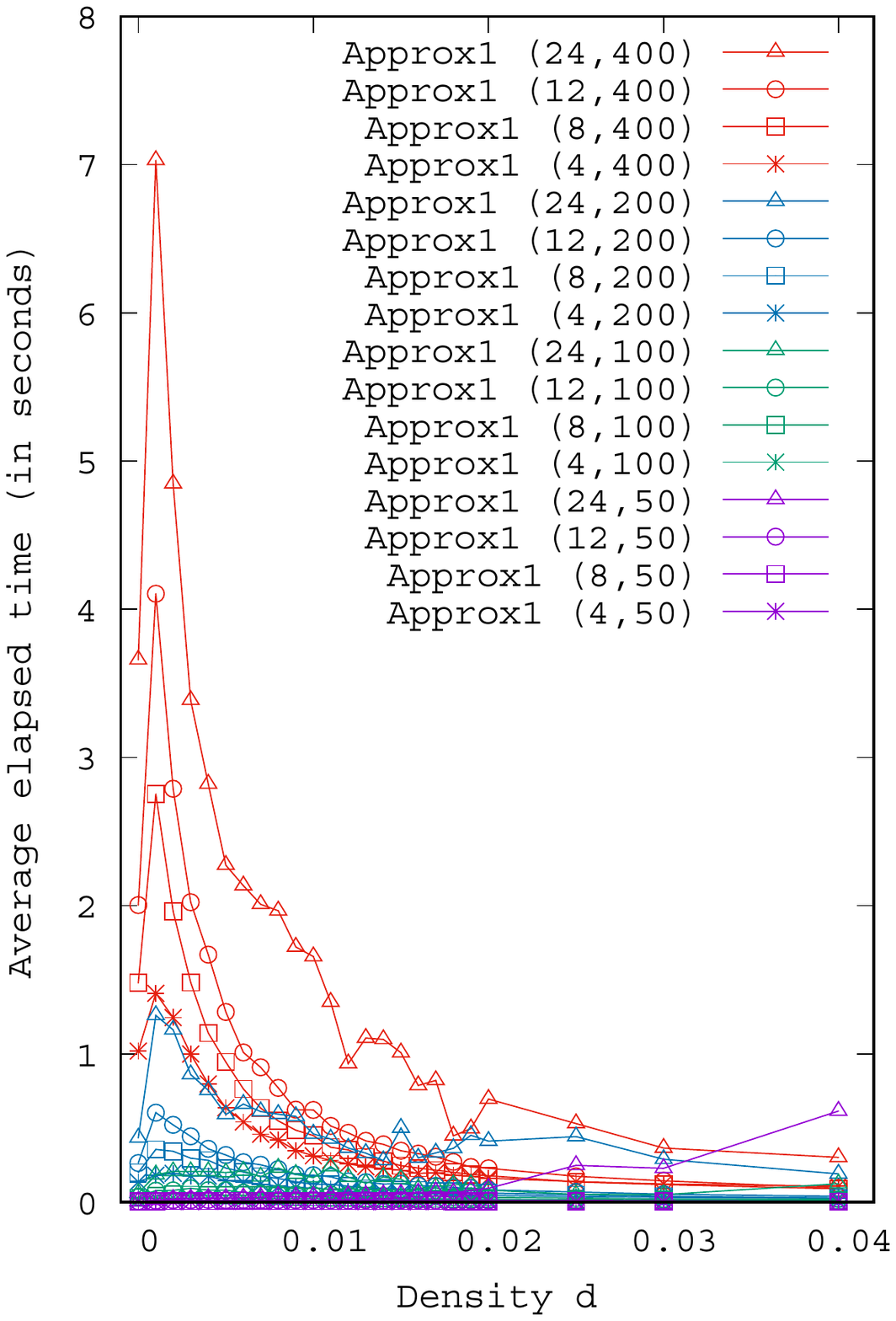}%
\mbox{\hspace{-3.2in}}\includegraphics[width=1.0\linewidth]{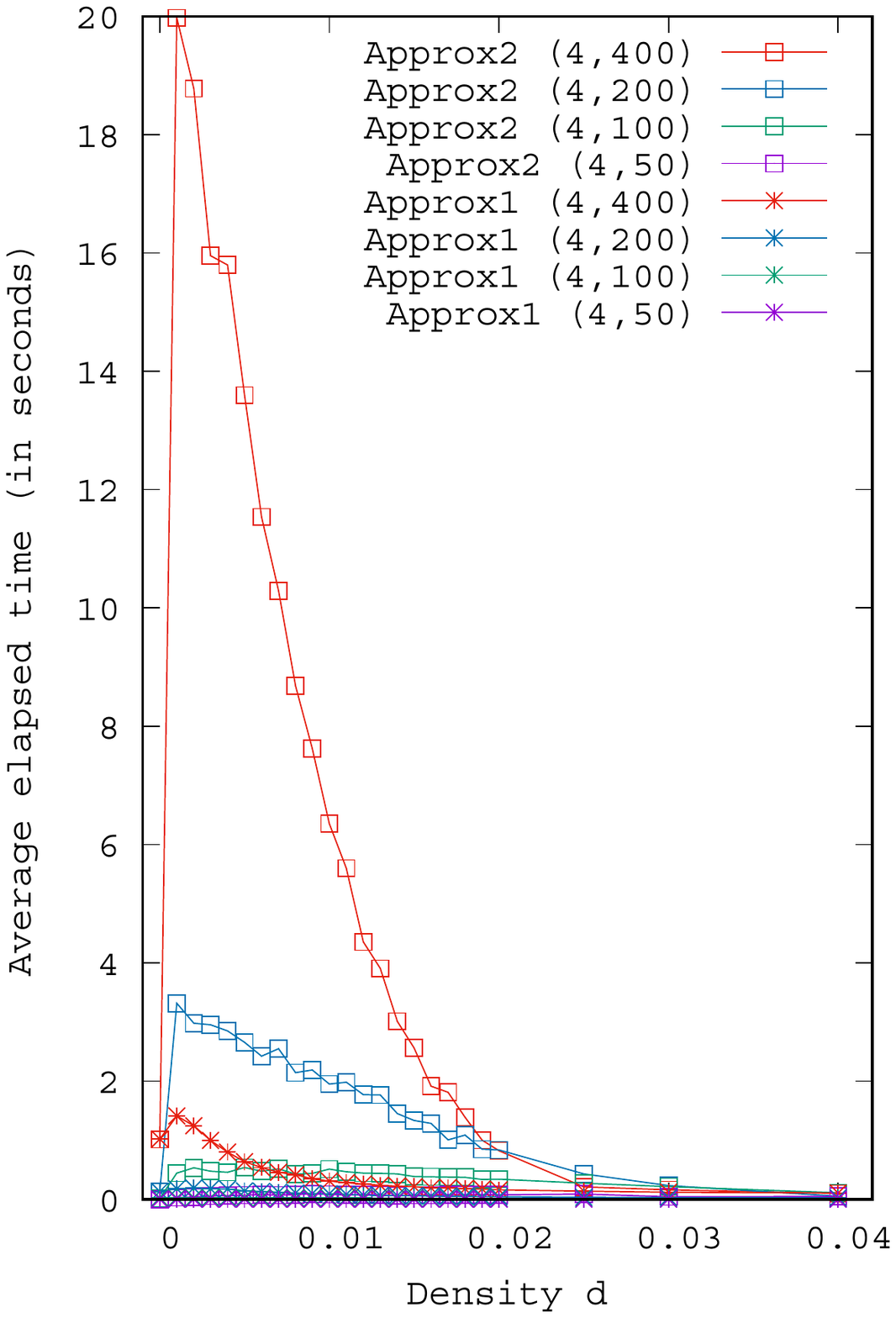}
\caption{The average running times of {\sc Approx1} (left), and of {\sc Approx2} when $k = 4$ (right),
	on the $100$ instance graphs generated with the triple $(k, n, d)$.\label{fig09}}
\end{figure}

\section{Conclusion}
In this paper, we studied the general vertex covering problem \Max{k}, where $k \ge 4$,
to find a collection of vertex-disjoint paths of order at least $k$ to cover the most vertices in the input graph.
The problem seemingly escapes from the literature,
but it admits a $k$-approximation algorithm by reducing to the maximum weighted $(2k-1)$-set packing problem~\cite{Neu21}.
We proposed the first direct $(0.4394 k + O(1))$-approximation algorithm
and an improved $2$-approximation algorithm when $k = 4$.
Both algorithms are local improvement based on a few operations, and we proved their approximation ratios via amortized analyses.
We have also designed numerical experiments to demonstrate the practical performance of the two algorithms,
and as expected they both performed well.
For example, across thousands of instances for \Max{4} ($800$ of them were reported),
they both ran very fast with the longest running time less than a minute on a typical computer,
and the observed worst performance ratio is $1.163$.

We suspect our amortized analyses are tight, and it would be interesting to either show the tightness 
or improve the analyses. 
For designing improved approximation algorithms,
one can look into whether the two new operations in {\sc Approx2} for $k = 4$ can be helpful for $k \ge 5$;
other different ideas might also work, for example,
one can investigate whether or not a maximum path-cycle cover~\cite{Gab83} can be taken advantage of.

On the other hand, we haven't addressed whether or not the \Max{k} problem, for a fixed $k \ge 4$, is APX-hard,
and if it is so, then it is worthwhile to show some non-trivial lower bounds on the approximation ratio, even only for $k = 4$.



\end{document}